\documentclass[11pt,a4paper]{article}
\usepackage[lmargin=1.0in,rmargin=1.0in,bottom=1.0in,top=1.0in,twoside=False]{geometry}

\usepackage{amsmath,amsthm}
\usepackage{mathtools}
\usepackage{graphicx}
\usepackage{subcaption}
\usepackage{enumerate}
\usepackage{tikz}
\usetikzlibrary{shapes}
\usetikzlibrary{arrows}
\usetikzlibrary{calc}

\tikzstyle{node}=[fill=black, shape=circle, inner sep=2pt]
\tikzstyle{node-is}=[fill={rgb,255: red,191; green,0; blue,64}, shape=circle, inner sep=2pt]
\tikzstyle{none}=[inner sep=0pt]
\tikzstyle{directed}=[->]
\tikzstyle{undirected}=[-]
\tikzstyle{bi-directional}=[<->]

\usepackage[T1]{fontenc}

\usepackage{enumitem}

\usepackage{microtype}
\usepackage{hyperref}
\usepackage{amsfonts}
\usepackage{comment}
\usepackage{mathrsfs}
\usepackage[capitalize, nameinlink]{cleveref}
\usepackage{todonotes}

\usepackage{stmaryrd}

\newtheorem{lemma}{Lemma}[section]
\newtheorem{theorem}[lemma]{Theorem}

\newtheorem{definition}[lemma]{Definition}
\newtheorem{claim}[lemma]{Claim}
\newtheorem{prop}[lemma]{Proposition}

\renewcommand{\mod}{\ \mathsf{mod}\ }
\newcommand{\Oh}{\mathcal{O}}
\newcommand{\sym}{\mathcal{S}}
\newcommand{\cA}{\mathcal{A}}
\newcommand{\cB}{\mathcal{B}}

\begin{document}
\iftrue
\author{
  Kacper Kluk%
  \thanks{Institute of Informatics, University of Warsaw, Poland. \texttt{k.kluk@uw.edu.pl}.}
  \and%
  Jesper Nederlof%
  \thanks{Department of Information and Computing Sciences, Utrecht University, The Netherlands. \texttt{j.nederlof@uu.nl}.}
}
\else
\author{}
\fi

\title{Lower bounds on pure dynamic programming for connectivity problems on graphs of bounded path-width
\iftrue
  \thanks{
  K.K. is supported by Polish National Science Centre SONATA BIS-12 grant number 2022/46/E/ST6/00143. J.N is supported by the project COALESCE that has received funding from the European Research Council (ERC), grant agreement No 853234.}\fi}

\date{}
\maketitle

\begin{abstract}
We give unconditional parameterized complexity lower bounds on pure dynamic programming algorithms -- as modeled by \emph{tropical circuits} -- for connectivity problems such as the Traveling Salesperson Problem. Our lower bounds are higher than the currently fastest algorithms that rely on algebra and give evidence that these algebraic aspects are unavoidable for competitive worst case running times. 

Specifically, we study input graphs with a small width parameter such as treewidth and pathwidth and show that for any $k$ there exists a graph $G$ of pathwidth at most $k$ and $k^{O(1)}$ vertices such that any tropical circuit calculating the optimal value of a Traveling Salesperson round tour uses at least $2^{\Omega(k \log \log k)}$ gates. We establish this result by linking tropical circuit complexity to the nondeterministic communication complexity of specific compatibility matrices. These matrices encode whether two partial solutions combine into a full solution, and Raz and Spieker~[Combinatorica 1995] previously proved a lower bound for this complexity measure.
\end{abstract}

\section{Introduction}
A common paradigm in theoretical computer science is to pin down which algorithmic techniques can or cannot achieve certain central goals by formalizing such a technique within a precise algorithmic model and then proving sharp limitations for that model. Classic examples include unconditional lower bounds for Linear Programming-based algorithms (i.e. extension complexity~\cite{FioriniMPTW15,Rothvoss17}), conditional lower bounds for preprocessing algorithms in parameterized complexity (conditioned on the non-collapse of the polynomial hierarchy~\cite{BodlaenderDFH09,FortnowS11}), and unconditional lower bounds for resolution-based algorithms for refuting the Strong Exponential Time Hypothesis~\cite{BeckI13}. This paradigm reveals which algorithmic nuances are fundamentally unavoidable and shows us where genuine breakthroughs must come from.

One notable example of such algorithmic nuance is that of algebraic cancellation: for several combinatorial computational problems, the fastest known algorithms crucially exploit certain algebraic cancellation in rather counter-intuitive ways, and it is a well-studied open question whether these algorithms can be matched with (more natural, and in some sense robust) `combinatorial' algorithms. This question arises for instance in the settings of Boolean matrix multiplication~\cite{WilliamsW18}, various variants of the problem of finding a perfect matching solvable by reduction to determinant computation~\cite{KarpUW86,MulmuleyVV87}, and fast exponential time algorithms for Hamiltonicity~\cite{Bjorklund14}. This very question of whether the use of algebraic cancellation is needed does not only occur in algorithm design, but also in various other disciplines. For example, in extremal combinatorics, it is an open question in various settings whether combinatorial proofs exist that match algebraic proofs (for example, a combinatorial proof of the skew two families theorem remains illusive~\cite{calbet2023k_r}), and the log rank conjecture in communication complexity asks whether algebraic low rank decompositions of matrices in general imply low rank combinatorial decompositions (or more specifically, whether the rank and partioning number of a Boolean matrix are quasi-polynomially related), see e.g.~\cite{LovettA14}.

One such setting in which we would like to replace algebraic arguments with combinatorial arguments within the area of parameterized complexity is that of connectivity problems parameterized by treewidth: In 2011 it was shown that that many connectivity problems parameterized by treewidth can be solved faster than what was deemed possible at the time~\cite{CyganNPPRW22}. In particular, problems like Hamiltonian Cycle and Steiner tree were solved with a randomized algorithm in $2^{O(k)}N^{O(1)}$ time, when given a tree decomposition of width $k$ of an $N$-vertex input graph. Before this work, fastest known algorithm for problems parameterized by treewidth were based on straightforward combinatorial dynamic programming techniques that naturally lead to running times like $k^{O(k)}N^{O(1)}$ and seemed hard to improve. Later work that followed up on~\cite{CyganNPPRW22} and provided deterministic algorithms for weighted extensions such as the Traveling Salesperson Problem (TSP) on bounded treewidth graphs~\cite{BodlaenderCKN15} and an alternative version connected the approach with a matroid-based extension of the aforementioned two-family theorem~\cite{FominLPS16}.

All these works~\cite{CyganNPPRW22,BodlaenderCKN15,FominLPS16} rely on algebraic cancellation by exploiting that the rank of certain \emph{compatibility} matrices is small, a particularly counterintuitive approach in comparison to the previous simple and clean combinatorial $k^{O(k)}N^{O(1)}$ time algorithms. It is a notable open question to remove the mentioned overhead in the runtime for the deterministic algorithm and weighted extensions. For example, a positive resolution of this open question would bring us closer to solving the TSP in time $1.9999^N$~\cite{Nederlof20}, where $N$ is the number of cities.

But, to do so it seems crucial to get a better combinatorial understanding of the involved compatibility matrices, simultaneously bypassing the undesired algebraic cancellation arguments. In this work we give a strong indication that more than just elementary combinatorial techniques are needed, via the algorithmic model of \emph{tropical circuits}.

\paragraph{Tropical Circuits.} A tropical circuit is an arithmetic circuit in which the inputs are labeled with variables that take an integer as value and the two arithmetic gates correspond to the max and sum, or alternatively, min and sum, operations. See \Cref{sec:prel} for a precise definition. The motivation for studying the expressiveness of tropical circuits is that it models a broad class dynamic programming algorithm. In particular, if there exists for a maximization problem (for which its instances are defined by a set of input weights) a dynamic programming algorithm that is `pure' in the sense that the associated recurrence only features the operations max and sum, than there is tropical circuit that outputs the optimal objective value with circuit size being proportional to the running time of the dynamic programming algorithm.

For example the classic Bellman-Held-Karp~\cite{Bellman62,heldKarp} $O(N^22^N)$ time algorithm for TSP is naturally converted into a tropical circuit with $O(N^22^N)$ gates and the $O(N^3)$ time Floyd-Warshall~\cite{Floyd62a,Warshall62} algorithm that computes the shortest path lengths between each pair of vertices is naturally converted into a tropical circuit with $O(N^3)$ gates. Both tropical circuits cannot be substantially improved~\cite{JerrumS82,kerr1970effect} (see also ~\cite[Corollary~2.2]{Jukna2023}).

In the realm of parameterized complexity, it is easy to see that canonical applications of dynamic programming can be modeled in an efficient way as tropical circuits. Examples are the dynamic programming algorithms for Steiner Tree and Set Cover (see e.g.~\cite[Section 6.1]{CyganFKLMPPS15}) with few number of terminals and elements and, especially relevant for this paper, the $O(2^{k} N)$ time algorithm for the maximum independent set problem on graphs with a given tree decomposition of width $k$ and $N$ vertices.

For much more detail on tropical circuits, we refer to the excellent textbook by Jukna~\cite{Jukna2023}.

\paragraph{Problems parameterized by treewidth.}
A very popular research line that started in~\cite{LokshtanovMS18b,LokshtanovMS18a} is that of investigating the fine-grained complexity of various NP-hard problems parameterized by width measures such as the treewidth of the input graph. 
In particular, for many NP-hard problems we are now able to design algorithms with a running time of the type $f(k)N^{O(1)}$ for some width measure $k$ and $n$ denoting the number of vertices of the input graph and simultaneously can prove that any improvement of this running time to $f(k)^{1-\Omega(1)}N^{O(1)}$ or even $f(k)^{o(1)}N^{O(1)}$ violates the Exponential Time Hypothesis (ETH) or the Strong Exponential Time Hypothesis (SETH). Such algorithms are often called \emph{(S)ETH-tight algorithms}. Such tight algorithms provide insight on how amenable the problem at hand is for divide and conquer algorithms since, conditioned on standard hypotheses, they reveal how much information of partial subsolutions exactly is relevant.

The aforementioned type of \emph{connectivity problems} such as TSP and Steiner Tree forms an important class of problems for which we do not generally have SETH-tight algorithms, precisely because the employed algebraic algorithms form a bottleneck towards deterministic algorithms and extensions to weighted variants that seemingly should be replaced with combinatorial arguments.

\paragraph{Our results.}
In this paper we provide evidence that direct combinatorial techniques on their own are insufficient for designing faster algorithm for connectivity problems, by giving unconditional lower bounds for tropical circuits. We state our lower bounds in terms of the~\emph{pathwidth} of the input graph. This is similar to the treewidth of a graph, except that we require more specifically to decompose the graph in a path-like manner instead of a tree-like manner. Hence, the pathwidth of a graph is always at least the treewidth of a graph and hence our lower bounds also imply lower bounds parameterized by treewidth. See \Cref{sec:prel} for definitions.
Before we study connectivity problems, we first study the complexity of a more basic problem:

\subparagraph{Maximum Weight Independent Set.}
In the Maximum Weight Independent Set problem one is given a graph $G=(V,E)$ along with a vertex weight $x_v$ for every $v \in V$, and is asked for the value $IS_G:=\max_{I \in \mathcal{I}(G)}\sum_{v \in I}x_v$ where $\mathcal{I}(G)$ denotes the family of all independent sets of $G$. Our lower bound for this problem reads as follows (full definitions are postponed to \Cref{sec:prel}):

\begin{theorem}\label{thm:main_is}
    For any $k \geq 1$, there exists a graph $G$ of pathwidth at most $k$ on $k^{O(1)}$ vertices such that any tropical circuit calculating $IS_G$ uses at least $\Omega(2^k)$ gates.
\end{theorem}
Note that pathwidth is a width measure that is always at least the treewidth of the graph.
Since there is a simple $O(2^{k}N)$-sized tropical circuit that calculates $IS_G$ for an $N$-vertex graph of pathwidth $k$, this result is optimal in a tight sense.
A previous result by Korhonen~\cite{Korhonen21} showed that for \emph{every} graph $G$ of treewidth $k$ and maximum degree $d$ any tropical circuit calculating $IS_G$ must be of size at least $2^{\Omega(k/d)}$.
This result is less tight than our new result and also seems to crucially rely on some properties of the Maximum Weight Independent Set problem. While the type of universal lower bound from~\cite{Korhonen21} is quite interesting, it does not directly have added value in our context of worst case complexity analysis and hence we do not pursue it further in this work.

\Cref{thm:main_is} is obtained using the following two ingredients. The graphs for which we show the bounds is constructed based on ideas from the classic reduction from CNF-Sat to Maximum Weight Independent Set. The tropical circuit size bound itself is shown by analyzing the structure of the so-called rectangles, a combinatorial notion that typically arises in the studie of tropical circuits that captures the way in which the partial solution calculated at some node of the tropical circuit can combine with the computations done by the rest of the circuit, see e.g.~\cite{Jukna2023}.

\subparagraph{Connectivity Problems.}
We show lower bounds on the tropical circuit complexity of the following graph connectivity problems. Let $G$ be an $N$-vertex graph with for every pair of distinct vertices $u$ and $v$ an edge weight $x_{u,v}\in \mathbb{N}$. Let $\mathcal{H}(G)$ denote the family of all sequences of $N + 1$ vertices $(u_i)_{i =0}^N$ such that $u_0 = u_N$ and $u_0 \to u_1 \to \dots \to u_{N - 1} \to u_N$ is a directed Hamiltonian cycle in $G$.
We define $DTSP_G$ to be the minimum $\sum_{i\in [N]}x_{u_{i-1},u_i}$ taken over all $(u_0,\ldots,u_N) \in \mathcal{H}(G)$. For an undirected graph, we define the undirected variant $TSP_G$ similarly while identifying variables $x_{u,v}$ and $x_{v, u}$ for all $uv \in E(G)$.

With these definitions in place, our main results can be stated as follows:
\begin{theorem}\label{thm:main_dtsp}
    For any $k \geq 1$, there exists a graph $G$ of pathwidth at most $k$ on $k^{O(1)}$ vertices such that any tropical circuit calculating $DTSP_G$ uses at least $2^{\Omega(k \log \log k)}$ gates.
\end{theorem}

\begin{theorem}\label{thm:main_tsp}
    For any $k \geq 1$, there exists a graph $G$ of pathwidth at most $k$ on $k^{O(1)}$ vertices such that any tropical circuit calculating $TSP_{G}$ uses at least $2^{\Omega(k \log \log k)}$ gates.
\end{theorem}
We also provide a similar lower bound for the Directed Spanning Tree problem.
We let $\mathcal{T}(G)$ denote the family of all functions $p : V(G) \to V(G)$ such that for exactly one vertex $v \in V(G)$, $p(v) = v$, and edges $(p(u), u)$ over all other vertices $u$ form an out-tree rooted at $v$ (i.e. a tree oriented away from the root). Furthermore, we define $DST_G$ as the minimum $\sum_{v \in V(G) : p(v) \neq v} x_{p(v), v}$ taken over all $p \in \mathcal{T}(G)$ (i.e., $DST_G$ is the minimum weight spanning tree of $G$).

\begin{theorem}\label{thm:main_dst}
    For any $k \geq 1$, there exists a graph $G$ of pathwidth at most $k$ on $k^{O(1)}$ vertices such that any tropical circuit calculating $DST_{G}$ uses at least $2^{\Omega(k \log \log k)}$ gates.
\end{theorem}

By adapting textbook dynamic programming algorithms, one can directly obtain tropical circuits for $DTSP_G,TSP_G$ and $DST_G$ of size $k^{O(k)} N$ if $G$ has $N$ vertices and pathwith/treewidth $k$. Thus there is still a gap between the lower bound and upper bound, and in fact this seems closely related to a similar gap in the area of communication complexity (see Section~\ref{sec:con} for more details on this). Nevertheless, our results show that the currently fastest $2^{O(k)}N$ time algorithms for TSP and Directed Steiner Tree (which generalizes DST) cannot be matched with merely pure dynamic programming.

The proofs of these three results combine the basic ingredients of the proof of \Cref{thm:main_is} with the properties of the \emph{Matchings Compatibility Matrix}, defined as follows. The rows and columns of this binary matrix are both indexed by perfect matchings of a bipartite graph, and an entry in the matrix indicates whether the union of these two perfect matchings form a Hamiltonian cycle. This matrix has already been studied in~\cite{raz1995log} in the context of the log-rank conjecture in communication complexity, where a lower bound on the non-deterministic communication complexity and an upper bound of its rank was given. The rank of (a slight variant of) this matrix also turned out to be important for studying the complexity of the Hamiltonian Cycle and TSP problem, both parameterized by the path/treewidth of the input graph and parameterized by the number of vertices~\cite{CurticapeanLN18,CyganKN18,Nederlof20}.

Our main technical contribution is that the lower bound on the non-deterministic communication complexity of this matrix from~\cite{raz1995log} can be used to obtain the above lower bounds on the size of tropical circuits.

\paragraph{Organization.}
This paper is organized as follows:
In \cref{sec:prel} we provide the necessary definitions and preliminary tools used in the remainder of the paper. \Cref{sec:mwis} presents the proof of \cref{thm:main_is}.
\cref{sec:match} presents the aforementioned matchings compatibility matrix and the required results about its structure.
In \cref{sec:tsp} we prove \cref{thm:main_dtsp} and \cref{thm:main_tsp}.
In \cref{sec:st} we prove \cref{thm:main_dst}, and we provide some concluding remarks in \cref{sec:con}.

\section{Preliminaries}\label{sec:prel}
We use the following basic notation. For a positive integer $\ell$, we use $[\ell]$ to denote the set $\{1, 2, \dots, \ell\}$. For a graph $G$, we use $V(G)$ to denote the set of vertices of $G$, and $E(G)$ to denote the set of edges of $G$. For a graph $G$ and a subset $X \subseteq V(G)$, we write $G[X]$ to denote the subgraph of $G$ induced by $X$, i.e., a graph with vertices restricted to $X$ and all edges between vertices of $X$ preserved.

By $\sym_k$ we denote the group of all permutations on the set $[k]$. By $\bar\sym_k$ we denote the subset of $\sym_k$ of permutations that contain exactly one cycle.
By $\sym_{2k}^2$ we denote the subset of $\sym_{2k}$ of permutations that contain exactly $k$ cycles, each of size $2$ and by $\sym_{2k}^k$ we denote the subset of $\sym_{2k}$ of permutations that contain exactly $2$ cycles, each of size $k$.

A \emph{cycle type} of a permutation is defined as the multiset of the sizes of all its cycles.
A\emph{conjugation} of a permutation $\rho \in \sym_k$ by a permutation $\pi \in \sym_k$ is defined as $\pi^{-1} \rho \pi$, i.e. the composition of the inverse of $\pi$, $\rho$ and $\pi$. It is well known that permutation conjugation preserves its cycle type, and hence $\bar\sym_k, \sym_{2k}^2, \sym_{2k}^k$ are closed under taking conjugations.
We will need the following group-theoretic properties of permutations and their conjugations. All of those are either well known or easy to obtain with some basic calculations done on subgroups.

\begin{prop}\label{prop:sym_k}
    The following equalities hold for any $k \geq 1$ and for any maximal set $\mathcal{Z}_k \subseteq \sym_k$ of permutations of the same cycle type.
    \begin{enumerate}
        \item $|\bar\sym_k| = (k - 1)!$, $|\sym_{2k}^2| = (2k - 1)!!$, $|\sym_{2k}^k| = (2k - 1)! / k$,
        \item $\{ \pi^{-1}\rho\pi \mid \pi \in \sym_k \} = \mathcal{Z}_k$ for every $\rho \in \mathcal{Z}_k$,
        \item $|\{ \pi \in \sym_k \mid \pi^{-1} \rho_1 \pi = \rho_2 \}| = |\sym_k| / |\mathcal{Z}_k|$ for every $\rho_1, \rho_2 \in \mathcal{Z}_k$,
        \item $|\{ \rho_2 \in \sym_{2k}^2 \mid \rho_2 \rho_1 \in \sym_{2k}^k \}| = (2k - 2)!!$ for every $\rho_1 \in \sym_{2k}^2$.
    \end{enumerate}
\end{prop}

\paragraph{Pathwidth.}
Let $G$ be a graph (undirected or directed). A \emph{path decomposition} of $G$ is a sequence of subsets $\beta_1, \beta_2, \dots, \beta_\ell \subseteq V(G)$ called \emph{bags} such that:
\begin{itemize}
    \item for every edge $(u,v) \in E(G)$, there is some $i \in [\ell]$ such that $u,v \in \beta_i$, and
    \item for every vertex $v \in V(G)$, there exist $1 \leq i_1 \leq i_2 \leq \ell$ such that $\{ i \mid v \in \beta_i \} = \{i_1, i_1 + 1, \dots, i_2 - 1, i_2\}$, i.e., the bags containing $v$ form a connected subinterval of the sequence of all bags.
\end{itemize}
The \emph{width} of a path decomposition $(\beta_i)_{i=1}^\ell$ is defined as $\max_{i \in [\ell]} |\beta_i| - 1$. The \emph{pathwidth} of $G$ is defined as the minimum possible width of a path decomposition of $G$.

\subsection{Tropical polynomials}
We treat all polynomials as defined over the tropical $(\max, +)$-semiring, i.e., $(f \cdot g)(x) = f(x) + g(x)$ and $(f + g)(x) = \max(f(x), g(x))$ for any two polynomials $f, g$.

\paragraph{Independent Set Polynomial.} If we treat the weights $x_v$ as indeterminates we can view $IS_G$ as a polynomial in the (max,+) semiring by replacing the max operation by addition and the $+$ operation by multiplication:
\begin{equation}\label{eq:mwispoly}
    IS_G := \sum_{I \in \mathcal{I}(G)} \left(\prod_{v \in I} x_v\right),
\end{equation}
where we remind the reader that $\mathcal{I}(G)$ denotes the family of all independent sets of $G$.

\paragraph{Traveling Salesperson Problem polynomial.} Similarly to the Independent Set polynomial $IS_G$, we define the directed TSP polynomial of $G$ as
$$
DTSP_G := \sum_{u_0, \dots, u_N \in \mathcal{H}(G)}
\left( \prod_{i \in [N]} x_{u_{i - 1}, u_i} \right),
$$
where we remind the reader that $\mathcal{H}(G)$ denotes the family of all sequences of $N + 1$ vertices $(u_i)_{i =0}^N$ such that $u_0 = u_N$ and $u_0 \to u_1 \to \dots \to u_{N - 1} \to u_N$ is a directed Hamiltonian cycle in $G$.
Naturally, we also view $TSP_G$ similarly as a polynomial by identifying variables $x_{u,v}$ and $x_{v, u}$ for all $uv \in E(G)$.

\paragraph{Directed Spanning Tree polynomial.} We also define a directed spanning tree polynomial of $G$ as
$$
DST_G := \sum_{p \in \mathcal{T}(G)}
\left( \prod_{v \in V(G) : p(v) \neq v} x_{p(v), v} \right)
$$
where we remind the reader that $\mathcal{T}(G)$ denotes the family of all functions $p : V(G) \to V(G)$ such that for exactly one vertex $v \in V(G)$, $p(v) = v$, and edges $(p(u), u)$ over all other vertices $u$ form an out-tree rooted at $v$.
\newline

\noindent
The \emph{support} of a monomial $m = x_1^{a_1} \dots x_\ell^{a_\ell}$ is the set of variables $\{x_1, \dots, x_\ell\}$ and is denoted as $\sup(m)$. The \emph{support of a polynomial} $p = m_1 + \dots + m_\ell$ is the union of supports $\sup(m_i)$ over all $i \in [\ell]$.
We let $\sup(p)$ denote the support of polynomial $p$.
For a monomial $m$ and a polynomial $p$, we write $m \in p$ if $p$ treated as a formal expression is of the form $m + q$ for some polynomial $q$. We say that a polynomial $p$ is \emph{homogeneous} if and only if for some $d \in \mathbb{N}$, the degree of all monomials $m \in p$ is $d$. For a polynomial $p$, we write $|p|$ to denote the number of monomials $m$ such that $m \in p$.

A \emph{valuation} is a function which maps variables to $\mathbb{R}$.
A \emph{characteristic valuation $\chi_m$} of a monomial $m$ is a valuation such that $\chi_m(x) = 1$ if $x \in \sup(m)$ and $\chi_m(x) = -1$ otherwise. As~stated before, we evaluate polynomials in a $(\max, +)$-semiring, i.e., given a valuation $v$, the~monomial $x$ evaluates to $v(x)$, $(f + g)(v)$ evaluates to $\max(f(v), g(v))$ and $(f \cdot g)(v)$ evaluates to $f(v) + g(v)$.

For two polynomials $p, q$, we write $p \subseteq q$ if for every monomial $m \in p$, we have $m \in q$. We write $p \simeq q$ iff $p \subseteq q$ and $q \subseteq p$. In particular, $p \simeq q$ doesn't imply that $p = q$ (take e.g. $p = x$ and $q = x + x$).
The definition of $\simeq$ is motivated by the following observation.

\begin{prop}\label{prop:poly_simeq}
    For any two multilinear polynomials $f, g$, we have $f \simeq g$ if and only if for any valuation $v$, we have $f(v) = g(v)$.
\end{prop}

\begin{proof}
    The implication from $f \simeq g$ to $f(v) = g(v)$ is immediate. In the other direction, assume w.l.o.g. that $f \subseteq g$ does not hold, and hence there is some monomial $m \in f$ such that $m \not\in g$.
    It is easy to see that $f(\chi_m) = |\sup(m)|$ and $g(\chi_m) < |\sup(m)|$, which is a contradiction.
\end{proof}

\subsection{Tropical circuits}

A \emph{tropical circuit} is a directed acyclic graph in which every vertex (called \emph{node}) is of in-degree 0 or 2. Vertices with in-degree 0 are labeled with either a variable, in which case they are called \emph{input nodes}, or with a constant $0$, in which case they are called \emph{constant nodes}. The nodes with in-degree $2$ are called \emph{operation nodes} and are labeled with a binary operator, either $+$ or $\max$. There is one node, designated an \emph{output node}.

Evaluation of a tropical circuit given a valuation $v$ is defined the following way. The nodes are processed according to the topological order of the graph. Constant nodes evaluate to value $0$. Input nodes labeled $x$ evaluate to $v(x)$. Operation nodes labeled with, respectively, $+$ and $\max$ evaluate to, respectively, sum and maximum of the values of its two predecessors. The output of the evaluation of the whole circuit is the value obtained at the output node.

We say that a tropical circuit $\Gamma$ \emph{calculates some polynomial $p$}, if the evaluation of $\Gamma$ on $v$ is equal to $p(v)$ for every valuation of variables $v$. It is easy to see that semantically, evaluating a tropical circuit computes exactly some polynomial: input node labeled $x$ computes $x$, the $\max$ nodes compute the sum of two polynomials and the $+$ nodes compute the product of two polynomials.

By \Cref{prop:poly_simeq} and the definitions, we immediately get the following.

\begin{prop}\label{prop:tc_semantics}
    Let $f, g$ be two multilinear polynomials and let $\Gamma$ be a tropical circuit calculating~$f$. Then $\Gamma$ calculates $g$ iff $f \simeq g$.
\end{prop}
Intuitively, this claim states that a tropical circuit calculates some polynomial iff the set of its monomials is exactly the set of all monomials which appear during the evaluation of said circuit.

The main combinatorial ingredient regarding tropical circuits we will use is the following decomposition lemma. For the proof, see, e.g., \cite[Lemma 3.4]{Jukna2023}.
\begin{lemma}\label{lem:tc_decomposition}
    Let $f$ be a homogeneous polynomial calculated by a tropical circuit of size $\tau$ and let $X$ denote an arbitrary subset of $\sup(f)$. Then $f$ can be written as
    $$
        f \simeq \sum_{i \in [\tau]} g_i \cdot h_i,
    $$
    where $|\sup(g_i) \cap X|, |\sup(h_i) \cap X| \leq \frac{2}{3}|X|$ for each $i \in [\tau]$.
\end{lemma}

\section{Independent Set}
\label{sec:mwis}

In this section, we show the bound of \Cref{thm:main_is}. The class of graphs we use are inspired by the classic reduction from CNF-sat to Maximum Weight Independent Set.
Fix $k > 0$ and let $q = 4\binom{k}{2}$. Consider all possible 2-CNF clauses on $k$ variables: there are exactly $q$ of those. We can think of each clause as a tuple $(a, b, n^a, n^b)$ where $a, b \in [k]$ represents an unordered pair of variables and $n^a, n^b \in \{0, 1\}$ represents whether $a$-th and $b$-th variable is negated in the clause. We define the undirected graph $G_k$ in the following way. We put
\begin{equation*}
    V(G_k) = \{v_{i,j} \mid i \in [k], j \in \{0, \dots, 2q - 1\}\} \cup
             \{ w_j \mid j \in [q] \}.
\end{equation*}
The vertices $v_{i,j}$ represent the literals of the CNF formula: $i$ denotes the variable and the parity of $j$ determines whether it is negated or not.
Every literal has $q$ copies, one for each clause. The vertices $w_j$ represent all possible 2-CNF clauses. For the sake of convenience, we will use $w_{a,b,n^a,n^b}$ to denote the vertex corresponding to the clause represented by $a,b,n^a,n^b$ as described above. We put $V_{i,r} = \{ v_{i, 2j + r} \mid j \in \{0, \dots q - 1\} \}$ for $i \in [k], r \in \{0, 1\}$, that is, even and odd vertices representing the $i$-th variable. We also put $V_i = V_{i,0} \cup V_{i,1}$.

We put the edges accordingly:
\begin{equation*}
    E(G_k) = \{ (v_{i,j-1}, v_{i,j}) \mid i \in [k], j \in [2q - 1]\} \cup
             \{ (w_i, l_{i,j} \mid i \in [q], j \in \{1, 2\} \}
\end{equation*}
where $l_{i,1} = v_{a_i, 2(i - 1) + n^a_i}$ and $l_{i,2} = v_{b_i, 2(i - 1) + n^b_i}$ where $w_i = w_{a_i, b_i, n^a_i, n^b_i}$. That is, $l_{i,j}$ correspond to the literals of the clause represented by $w_i$. Because there is a copy $v_{i,j}$ of each literal for every clause, every $v_{i,j}$ is connected at most to one vertex $w_{1 + \lfloor j/2 \rfloor}$. The degree of each $w_i$ is exactly $2$.

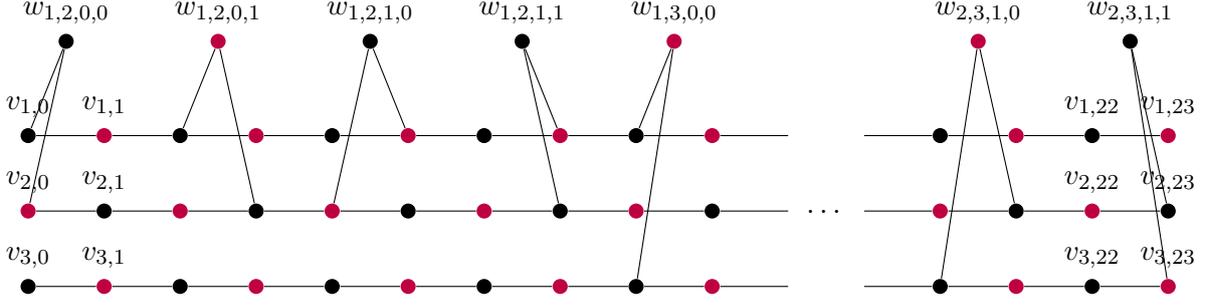
\begin{figure}[h]
    \centering
    \label{fig:mwis}
    \begin{tikzpicture}
		\node [style=node] (0) at (-6, 1) [label={$v_{1,0}$}] {};
		\node [style=node-is] (1) at (-5, 1) [label={$v_{1,1}$}] {};
		\node [style=node] (2) at (-4, 1) {};
		\node [style=node-is] (3) at (-3, 1) {};
		\node [style=node] (4) at (-2, 1) {};
		\node [style=node-is] (5) at (-1, 1) {};
		\node [style=node] (6) at (0, 1) {};
		\node [style=node-is] (7) at (1, 1) {};
		\node [style=node] (8) at (2, 1) {};
		\node [style=node-is] (9) at (3, 1) {};
		\node [style=node-is] (10) at (-6, 0) [label={$v_{2,0}$}] {};
		\node [style=node] (11) at (-5, 0) [label={$v_{2,1}$}] {};
		\node [style=node-is] (12) at (-4, 0) {};
		\node [style=node] (13) at (-3, 0) {};
		\node [style=node-is] (14) at (-2, 0) {};
		\node [style=node] (15) at (-1, 0) {};
		\node [style=node-is] (16) at (0, 0) {};
		\node [style=node] (17) at (1, 0) {};
		\node [style=node-is] (18) at (2, 0) {};
		\node [style=node] (19) at (3, 0) {};
		\node [style=node] (20) at (-6, -1) [label={$v_{3,0}$}] {};
		\node [style=node-is] (21) at (-5, -1) [label={$v_{3,1}$}] {};
		\node [style=node] (22) at (-4, -1) {};
		\node [style=node-is] (23) at (-3, -1) {};
		\node [style=node] (24) at (-2, -1) {};
		\node [style=node-is] (25) at (-1, -1) {};
		\node [style=node] (26) at (0, -1) {};
		\node [style=node-is] (27) at (1, -1) {};
		\node [style=node] (28) at (2, -1) {};
		\node [style=node-is] (29) at (3, -1) {};
		\node [style=node] (30) at (-5.5, 2.25) [label={$w_{1,2,0,0}$}] {};
		\node [style=node-is] (31) at (-3.5, 2.25) [label={$w_{1,2,0,1}$}] {};
		\node [style=node] (32) at (-1.5, 2.25) [label={$w_{1,2,1,0}$}] {};
		\node [style=node] (33) at (0.5, 2.25) [label={$w_{1,2,1,1}$}] {};
		\node [style=node-is] (34) at (2.5, 2.25) [label={$w_{1,3,0,0}$}] {};
		\node [style=none] (35) at (4, 1) {};
		\node [style=none] (36) at (4, 0) {};
		\node [style=none] (37) at (4, -1) {};
		\node [style=none] (38) at (4.5, 0) {$\dots$};
		\node [style=none] (39) at (5, 1) {};
		\node [style=none] (40) at (5, 0) {};
		\node [style=none] (41) at (5, -1) {};
		\node [style=node] (42) at (6, 1) {};
		\node [style=node-is] (43) at (7, 1) {};
		\node [style=node] (44) at (7, 0) {};
		\node [style=node-is] (45) at (6, 0) {};
		\node [style=node] (46) at (6, -1) {};
		\node [style=node-is] (47) at (7, -1) {};
		\node [style=node-is] (48) at (6.5, 2.25) [label={$w_{2,3,1,0}$}] {};
		\node [style=node] (49) at (8, 1) [label={$v_{1,22}$}] {};
		\node [style=node-is] (50) at (9, 1) [label={$v_{1,23}$}] {};
		\node [style=node] (51) at (9, 0) [label={$v_{2,23}$}] {};
		\node [style=node-is] (52) at (8, 0) [label={$v_{2,22}$}] {};
		\node [style=node] (53) at (8, -1) [label={$v_{3,22}$}] {};
		\node [style=node-is] (54) at (9, -1) [label={$v_{3,23}$}] {};
		\node [style=node] (55) at (8.5, 2.25) [label={$w_{2,3,1,1}$}] {};
        
		\draw [style=undirected] (0) to (1);
		\draw [style=undirected] (1) to (2);
		\draw [style=undirected] (2) to (3);
		\draw [style=undirected] (3) to (4);
		\draw [style=undirected] (4) to (5);
		\draw [style=undirected] (5) to (6);
		\draw [style=undirected] (6) to (7);
		\draw [style=undirected] (7) to (8);
		\draw [style=undirected] (8) to (9);
		\draw [style=undirected] (10) to (11);
		\draw [style=undirected] (11) to (12);
		\draw [style=undirected] (12) to (13);
		\draw [style=undirected] (13) to (14);
		\draw [style=undirected] (14) to (15);
		\draw [style=undirected] (15) to (16);
		\draw [style=undirected] (16) to (17);
		\draw [style=undirected] (17) to (18);
		\draw [style=undirected] (18) to (19);
		\draw [style=undirected] (20) to (21);
		\draw [style=undirected] (21) to (22);
		\draw [style=undirected] (22) to (23);
		\draw [style=undirected] (23) to (24);
		\draw [style=undirected] (24) to (25);
		\draw [style=undirected] (25) to (26);
		\draw [style=undirected] (26) to (27);
		\draw [style=undirected] (27) to (28);
		\draw [style=undirected] (28) to (29);
		\draw [style=undirected] (0) to (30);
		\draw [style=undirected] (10) to (30);
		\draw [style=undirected] (31) to (2);
		\draw [style=undirected] (31) to (13);
		\draw [style=undirected] (32) to (14);
		\draw [style=undirected] (32) to (5);
		\draw [style=undirected] (33) to (7);
		\draw [style=undirected] (33) to (17);
		\draw [style=undirected] (34) to (8);
		\draw [style=undirected] (34) to (28);
		\draw [style=undirected] (9) to (35.center);
		\draw [style=undirected] (19) to (36.center);
		\draw [style=undirected] (29) to (37.center);
		\draw [style=undirected] (39.center) to (42);
		\draw [style=undirected] (40.center) to (45);
		\draw [style=undirected] (41.center) to (46);
		\draw [style=undirected] (42) to (43);
		\draw [style=undirected] (45) to (44);
		\draw [style=undirected] (46) to (47);
		\draw [style=undirected] (55) to (51);
		\draw [style=undirected] (55) to (54);
		\draw [style=undirected] (49) to (50);
		\draw [style=undirected] (52) to (51);
		\draw [style=undirected] (53) to (54);
		\draw [style=undirected] (48) to (44);
		\draw [style=undirected] (48) to (46);
		\draw [style=undirected] (43) to (49);
		\draw [style=undirected] (44) to (52);
		\draw [style=undirected] (47) to (53);
\end{tikzpicture}
    \caption{The graph $G_3$. Red vertices belong to the canonical solution for valuation $\rho(1) = 1, \rho(2) = 0, \rho(3) = 1$.}
\end{figure}

It is easy to see that the graph $G_k$ has pathwidth at most $k + 1$: We can start with bags $\{v_{i,j}:i \in [k]\}$ for $j \in \{0,\ldots,2q-1\}$ and add the sets $N[w_j]$ to them in the natural way. Proving the following lemma will immediately show \Cref{thm:main_is}.

\begin{lemma}\label{lem:mwis}
	For any $k \geq 1$, any tropical circuit calculating $IS_{G_k}$ uses at least $2^k / 3$ gates.
\end{lemma}

The rest of this section is dedicated to proving this lemma. Fix $k > 0$ and put $G := G_k$.
The proof will be done in two steps. First, we will show a combinatorial structure of so-called rectangles in our described graph. Then, we will relate this notion to the polynomials calculated by tropical circuits and show the desired lower bound as a consequence.

The \emph{canonical solution $I_\rho$ given by the assignment $\rho : [k] \to \{0, 1\}$} is defined as
\begin{equation*}
	I_\rho = \{ v_{i, 2j + \rho(i)} \mid i \in [k], j \in \{0, \dots, q-1\} \} \cup
	         \{ w_i \mid \rho(a_i) \neq n^a_i \ \text{and}\ \rho(b_i) \neq n^b_i \}.
\end{equation*}
That is, we pick vertices $v_{i,j}$ for all $j$ even or all $j$ odd depending on $\rho(i)$ and include all vertices $w_i$ corresponding to clauses not satisfied by assignment $\rho$. Clearly, there are $2^k$ canonical solutions, one per each assignment $\rho$, and every canonical solution is an independent set of $G$.

We will say that a pair $\cA, \cB \subseteq 2^{V(G)}$ of families of independent sets of $G$ forms a \emph{rectangle} in $G$ (denoted $\cA \cdot \cB$) if and only if every pair of sets $A \in \cA, B \in \cB$ are disjoint and non-adjacent (in particular, $A \cup B$ forms an independent set in $G$). We refer to the families $\cA, \cB$ as \emph{sides} of the rectangle. We will say that a rectangle $\cA \cdot \cB$ \emph{contains} an independent set $I$ if $I = A \cup B$ for some $A \in \cA, B \in \cB$. We will say that a set $A \in \cA$ is \emph{useful} if there exists $B \in \cB$ such that $A \cup B$ is a canonical solution. We define a set $B \in \cB$ being useful in an analogous way.

For simplicity, the notion of a rectangle introduced here is defined in terms of families of independent sets of $G$. This concept, however, will be crucial in showing all of the subsequent bounds, hence, in later sections we will redefine rectangles in terms of tropical polynomials, as in the statement of \Cref{lem:tc_decomposition}, in order for the definition to be more general.

\subsection{Thin rectangles}

The key combinatorial property of the circuits calculating $IS_G$ is expressed via the following lemma.
\begin{lemma}\label{lem:is_thin_rects}
	Let $\cA \cdot \cB$ be a rectangle in $G$ that contains at least one canonical solution $I_\rho$. Then, either $\cA$ or $\cB$ contain at most one useful set.
\end{lemma}

For the rest of this subsection, we fix a rectangle $\cA \cdot \cB$, and focus on proving the lemma via the following series of claims.
\begin{claim}
    Let $I_1, I_2$ be two different useful sets belonging to the same side of $\cA \cdot \cB$. Then, there exists $i \in [k], j \in \{0, \dots, 2q - 1\}$ such that $v_{i, j}$ belongs to exactly one of $I_1, I_2$.
\end{claim}
\begin{proof}
    W.l.o.g. we assume that $I_1, I_2 \in \cA$. If both sets have different intersection with some $V_i$, then the claim trivially follows, hence w.l.o.g. we assume that for some $l \in [q]$, we have $w_l \in I_1$ and $w_l \not\in I_2$.
    
    Since $I_2$ is useful, we have $J_2 \in \cB$ such that $I_2 \cup J_2$ is a canonical solution. Since $I_1$ and $J_2$ are disjoint, we have $w_l \not\in I_2 \cup J_2$, hence $v_{i,j} \in I_2 \cup J_2$ for some neighbor $v_{i, j}$ of $w_l$. As $I_1$ and $J_2$ are non-adjacent, we have $v_{i,j} \not\in J_2$, hence $v_{i,j} \in I_2$. As $I_1$ is an independent set of $G$, we have $v_{i,j} \not\in I_1$, which finishes the proof of the claim.
\end{proof}

\begin{claim}
    Let $I_1, I_2$ be two different useful sets belonging to the same side of $\cA \cdot \cB$. Then, there exists $i \in [k], r \in \{0,1\}$ such that $V_{i,r} \subseteq I_1$ and $V_{i, 1-r} \subseteq I_2$.
\end{claim}
\begin{proof}
    Assume w.l.o.g. that $I_1, I_2 \in \cA$.
    By the previous claim, we have some $i \in [k], j \in \{0, \dots, 2q - 1\}$ such that w.l.o.g. $v_{i,j} \in I_1$ and $v_{i,j} \not\in I_2$. Put $r = j \mod 2$. Since $I_1, I_2$ are useful, we have the corresponding $J_1, J_2 \in \cB$. By contradiction, let $j' \in \{0, \dots, 2q - 1\}$ be an index minimizing $|j - j'|$ for which $v_{i,j'}$ does not belong to the expected set, i.e., such that either $j' \mod 2 = r$ and $v_{i,j'} \not\in I_1$ or $j' \mod 2 \neq r$ and $v_{i,j'} \not\in I_2$.

    Consider the case $j' \mod 2 = r$, the other one will be symmetric. As $v_{i,j} \in I_1 \cup J_1$, we also have $v_{i,j'} \in I_1 \cup J_1$. By minimality, we have that either $v_{i,j'+1}$ or $v_{i,j'-1}$ is in $I_2$. If $v_{i,j'} \not\in I_1$ then $v_{i,j'} \in J_1$, but this would imply that $I_2 \cup J_1$ is not an independent set, which is a contradiction.
\end{proof}

\begin{claim}
    Let $I_1, I_2 \in \cA$ be two different useful sets and let $J_1 \in \cB$ be such that $I_1 \cup J_1$ is a canonical solution. Then $J_1$ is the only useful set in $\cB$.
\end{claim}
\begin{proof}
    Assume by contradiction that we have useful $J_2 \in \cB$ different to $J_1$. Applying the previous claim to both sides, we obtain $i_A, i_B \in [k]$ and $r_A, r_B \in \{0,1\}$ such that: $V_{i_A,r_A} \subseteq I_1, V_{i_A, 1 - r_A} \subseteq I_2, V_{i_B, r_B} \subseteq J_1$ and $V_{i_B, 1 - r_B} \subseteq J_2$. Naturally, we have $i_A \neq i_B$.

    Pick $w = w_{i_A, i_B, 1 - r_A, 1 - r_B}$. That is,  for some $j \in \{0, \dots, q-1\}$, the neighbors of $w$ are exactly $v_A = v_{i_A, 2j + 1 - r_A}$ and $ v_B = v_{i_B, 2j + 1 - r_B}$. We have $v_A \in I_2$ and $v_B \in J_2$. Since $V_{i_A, r_A} \subseteq I_1$ and $I_1$ is independent in $G$, we have $v_A \not\in I_1$. Similarly, $v_B \not\in J_1$. Since $I_2$ and $J_1$ are disjoint, we have $v_A \not\in J_1$, and similarly $v_B \not\in I_1$. Since $I_2$ and $J_1$ are not adjacent, we have $w \not\in J_1$, and similarly $w \not\in I_1$. However, this means that $I_1 \cup J_1$ is not a canonical solution since neither $w$ nor its two neighbors belong to $I_1 \cup J_1$. This is a contradiction, hence the claim is proven.
\end{proof}
\Cref{lem:is_thin_rects} follows immediately from the last claim.

\subsection{Lower bound on thin rectangle circuits}
In this subsection we prove \Cref{thm:main_is}. At this point we can abstract away from the exact structure of $G$. The proof will depend only on the thin rectangle property proven in the previous subsection. The proof will closely follow the proof of \cite[Lemma 2.18]{Jukna2023}.

Let $\Gamma$ be any tropical circuit that calculates $IS_G$. For any node $w$ of $\Gamma$, we define its \emph{below} $B_w$ as the polynomial calculated by $\Gamma$ if we designate $w$ to be its output node. Intuitively, $B_w$ captures the contribution of the subcircuit rooted at $w$ to the output of the whole calculation.

In a similar spirit, we would like to define the \emph{above} of a node $w$ as the contribution of the remainder of the circuit. Note that $B_w$ does not need to be contained in $IS_G$, however, assuming the output node is reachable from $w$, for at least one monomial $m$, we have $m \cdot B_w \subseteq IS_G$. We will define $A_w$ as the sum of all such monomials, that is
$$
    A_w = \sum_{m:m \cdot B_w \subseteq IS_G} m.
$$
Naturally, $A_w \cdot B_w \subseteq IS_G$ for every node $w$ from which the output node is reachable.

Every monomial of each $A_w, B_w$ represents an independent set in $G$. Let $A_w^*$ denote the useful monomials of $A_w$, that is, monomials $m \in A_w^*$ such that there exist $m' \in B_w$ for which $m \cdot m'$ represents a canonical solution. We define $B_w^*$ analogously. As proven by \Cref{lem:is_thin_rects}, for every node $w$, one of $A_w^*, B_w^*$ must be of size at most $1$. If $|A_w^*| \leq 1$, we will say that $A_w \cdot B_w$ is \emph{$A$-thin}, otherwise, we will say that it is \emph{$B$-thin}.

If $w$ is an output node, then we have $A_w = \{1\}, B_w = IS_G$, so $A_w \cdot B_w$ is $A$-thin. If $w$ is an input node labeled $x$, we have $B_w = \{x\}$, therefore $A_w \cdot B_w$ is $B$-thin. Let $w_\ell, \dots, w_1$ be any path in $\Gamma$ where $w_1$ is an output node and $w_\ell$ is some input node. Based on our observations, there is some edge on this path $(w_{i+1}, w_i)$ such that $w_i$ is $A$-thin and $w_{i + 1}$ is $B$-thin, and so $|A_{w_i}^* \cdot B_{w_{i+1}}^*| \leq 1$.

Let $IS_G^*$ be the polynomial containing all monomials of $IS_G$ representing a canonical solution. The following claim shows that all canonical solutions belonging to the rectangles described above can be propagated along such input to output paths of $\Gamma$. We will use the term predecessor to refer to in-neighbors of a node of the circuit.

\begin{claim}
	Let $w$ be a node of $\Gamma$ and let $c \in IS_G^*$ be such that $c = a \cdot b$ for some $a \in A_w^*, b \in B_w^*$.
	Then, for at least one predecessor $u$ of $w$, we have $c \in A_u^* \cdot B_u^*$.
	Moreover, if $w$ is a + gate, then the above holds for both predecessors of $w$.
\end{claim}

\begin{proof}
	Let $u, v$ be the predecessors of $w$.
	First, consider the case where $w$ is a max gate, that is, $B_w = B_u + B_v$. Then, either $b \in B_u$ or $b \in B_v$. W.l.o.g. assume the former.
    We have $a \cdot B_u \subseteq a \cdot B_w \subseteq IS_G$, hence $a \in A_u$. Therefore, $b \in B_u^*$, $a \in A_u^*$, and hence $c \in A_u^* \cdot B_u^*$ as desired.
    
	Second, consider the case where $w$ is a + gate, that is, $B_w = B_u \cdot B_v$. Then, we have $b = b_v \cdot b_u$ for some $b_v \in B_v$ and $b_u \in B_u$.
    We have $a \cdot b_v \cdot B_u \subseteq a \cdot B_v \cdot B_u \subseteq IS_G$, hence $a \cdot b_v \in A_u$. Therefore, $b_u \in B_u^*$, $a \cdot b_v \in A_u^*$, and hence $c \in A_u^* \cdot B_u^*$. By a symmetric argument, $c \in A_v^* \cdot B_v^*$.
\end{proof}

Now, consider the following process. We fix a canonical solution $c \in IS_G^*$ and start at output node $w_1$. If we are currently in a max node $w_i$, we move towards the predecessor $w_{i+1}$ given by the claim. If we are currently in a + node, we move towards the predecessor $w_{i+1}$ for which $|B_{w_{i+1}}^*|$ is larger. Such process terminates at an input node and produces a path $w_\ell, \dots, w_1$. As argued before, for some $i \in [\ell - 1]$, we have $|A_{w_i}^* \cdot B_{w_{i+1}}^*| \leq 1$.

We have $c \in A_{w_{i+1}}^* \cdot B_{w_{i+1}}^*$. If $w_i$ is a max node, we have $B_{w_{i+1}} \subseteq B_{w_i}$, hence $A_{w_{i+1}}^* \subseteq A_{w_i}^*$, hence $c \in A_{w_i}^* \cdot B_{w_{i+1}}^*$, and hence $A_{w_i}^* \cdot B_{w_{i+1}}^* = \{c\}$.

If $w_i$ is a + node, then we additionally have $c \in A_v^* \cdot B_v^*$ where $v$ is the predecessor of $w_i$ other than $w_{i+1}$. Given the way we chose $w_{i+1}$, we have $|B_v^*| \leq |B_{w_{i+1}}^*| \leq 1$. Since $B_{w_i}^* \subseteq B_{w_{i+1}}^* \cdot B_v^*$, we have $|B_{w_i}^*| \leq 1$. Thus, $A_{w_i}^* \cdot B_{w_i}^* = \{c\}$.

Repeating this process for each canonical $c$ creates a mapping from $IS^*_G$ to the set $V(G) \cup E(G)$. If we map $c$ to a vertex $w$, then we have a guarantee that $A_w^* \cdot B_w^* = \{c\}$. If we map $c$ to an edge $(u,w)$, we have a guarantee that $A_w^* \cdot B_u^* = \{c\}$. This implies that the mapping is injective, and hence $2^k = |IS^*_G| \leq |V(G)| + |E(G)| \leq 3|V(G)|$. This finishes the proof of the lemma.

\section{Matching compatibility matrix}
\label{sec:match}

In the following section, a matrix is a function from $I \times J$ to an arbitrary value set, where $I$ and $J$ are some sets of indices of, respectively, rows and columns of the matrix. We do not require $I, J$ to be a set of form $[n]$ for some $n \in \mathbb{N}$. All matrices considered in this section have values in the set $\{0, 1\}$.

\begin{definition}
    For a 0-1 matrix $M$ with row indices $I$ and column indices $J$, we say that a~pair $I' \subseteq I, J' \subseteq J$ forms a \emph{rectangle} of $M$ if and only if it induces an all-ones submatrix of $M$, i.e., if and only if $M_{i,j} = 1$ for each $i \in I', j \in J'$. The \emph{size} of a rectangle $R$ is defined as $|I'| \cdot |J'|$ and denoted as $|R|$.
\end{definition}

\begin{definition}
    A \emph{rectangle cover} of a 0-1 matrix $M$ is a set of rectangles $(I_1, J_1), \dots, (I_s, J_s)$ which cover all ones of $M$, i.e., such that for each $i,j$ with $M_{i,j} = 1$, there exists $p \in [s]$ such that $(i, j) \in I_p \times J_p$. The \emph{size} of the cover is the number of rectangles $s$.
\end{definition}

\subsection{Complete bipartite graphs}

\begin{definition}
    A matching compatibility matrix $\mathcal{M}_k$ of order $k$ is a binary matrix of size $k! \times k!$ with rows and columns indexed by permutations $\sym_k$ which satisfies
    $$
        \mathcal{M}_k(\rho_1, \rho_2) = 1 \quad \text{if and only if} \quad
        \rho_2 \rho_1 \in \bar\sym_k.
    $$
    By $C_k$ we denote the size of the smallest rectangle cover of $\mathcal{M}_k$.
\end{definition}

Our bounds on sizes of tropical circuits will be based on the following bound on $C_k$ due to Raz and Spieker.
\begin{lemma}[\cite{raz1995log}]\label{lem:raz_spieker}
    $C_k = 2^{\Omega(k \log \log k)}$.
\end{lemma}
It is worth noting that $C_k$ is believed to be bounded by $2^{\Omega(k \log k)}$. Showing this would give us asymptotically tight bounds on the value of $C_k$ as the upper bound of $k! = 2^{\Oh(k \log k)}$ is trivial. This, however, remains an open problem.

A small rectangle cover of a matrix implies the existence of large rectangles in it, but the converse does not need to hold in the general case. The following claim and lemma shows, that in the case of $\mathcal{M}_k$, the size of a minimal rectangle cover and maximal rectangle are in fact related, and within a poly-logarithmic factor of what one can expect.

\begin{claim}\label{clm:mm_rect_cover_duality}
    Let $R$ be any rectangle in $\mathcal{M}_k$. Then, there exists a rectangle cover of $\mathcal{M}_k$ of size $\ell = \frac{k!(k - 1)!}{|R|} \cdot 2k \ln k$.
\end{claim}

\begin{proof}
    Let $Q \subseteq \sym_k \times \sym_k$ denote the set of pairs $(\rho_1, \rho_2)$ such that $\rho_2\rho_1 \in \bar\sym_k$, i.e., the set of 1-entries of $\mathcal{M}_k$. Let $P_1$ and $P_2$ denote the sets of permutations which are indices of, respectively, rows and columns of $R$. Thus, for every $\rho_1 \in P_1, \rho_2 \in P_2$, we have $(\rho_1, \rho_2) \in Q$.
    Consider a map $\mu_{\alpha, \beta} : \sym_k \times \sym_k \to \sym_k \times \sym_k$ parameterized by permutations $\alpha, \beta \in \sym_k$, defined as
    $$
        \mu_{\alpha, \beta} (\rho_1, \rho_2) =
        (\alpha \rho_1 \beta, \beta^{-1} \rho_2 \alpha^{-1}).
    $$
    First, note that $\rho_2 \rho_1 \in \bar\sym_k$ if and only if $(\beta^{-1}\rho_2\alpha^{-1}) (\alpha \rho_1 \beta) = \beta^{-1} \rho_2 \rho_1 \beta \in \bar\sym_k$, as permutation conjugation preserves cycle type, hence $(\rho_1, \rho_2) \in Q$ if and only if $\mu_{\alpha, \beta}(\rho_1, \rho_2) \in Q$.

    Let $\alpha, \beta$ be two random permutations sampled independently from a uniform distribution on~$\sym_k$. First, we show that for any two pairs $(\rho_1, \rho_2), (\sigma_1, \sigma_2) \in Q$, we have
    $$
    \Pr[\mu_{\alpha, \beta}(\rho_1, \rho_2) = (\sigma_1, \sigma_2)] \geq \frac{1}{k!(k - 1)!}.
    $$
    To do this, we show that there exist at least $k$ different pairs $\alpha, \beta$ for which $\alpha \rho_1 \beta = \sigma_1$ and $\beta^{-1} \rho_2 \alpha^{-1} = \sigma_2$. By \Cref{prop:sym_k}, we have exactly $k$ permutations $\beta$ which satisfy $\sigma_2 \sigma_1 = \beta^{-1} \rho_2 \rho_1 \beta$. For each of those, we can put $\alpha = \sigma_2^{-1} \beta^{-1} \rho_2$ to obtain a pair satisfying the conditions.

    For any pair $(\rho_1, \rho_2) \in Q$, we have $(\rho_1, \rho_2) \in \mu_{\alpha, \beta}(R)$ if and only if $\mu_{\alpha^{-1}, \beta^{-1}}(\rho_1, \rho_2) \in R$, and therefore
    $$
    \Pr[(\rho_1, \rho_2) \in \mu_{\alpha, \beta}(R)] \geq \frac{|R|}{k!(k - 1)!}.
    $$

    Now, sample $\ell$ pairs $\alpha_i, \beta_i$ independently uniformly from $\sym_k \times \sym_k$ and let $R_i = \mu_{\alpha_i, \beta_i}(R)$. Clearly $R_i \subseteq Q$ for each $i \in [\ell]$. Let $\bar{R} = \bigcup_{i \in [\ell]} R_i$. For any $(\rho_1, \rho_2) \in Q$, we have
    $$
    \Pr\left[(\rho_1, \rho_2) \not\in \bar{R}\right] \leq \left(1 - \frac{|R|}{k!(k - 1)!} \right)^{\ell} < e^{-2k \ln k} = k^{-2k} < (k!(k - 1)!)^{-1}.
    $$
    By union bound, we have
    $
    \Pr \left[ \bar{R} \neq Q \right] < 1
    $,
    hence there exist a choice of $\alpha_i, \beta_i$ for which $\bar{R} = Q$. For this choice of $\alpha_i,\beta_i$ we have that $R_1,\ldots,R_\ell$ is a rectangle cover, which finishes the proof.
\end{proof}

Immediately, we get the following lemma as a corollary.

\begin{lemma}\label{lem:mm_size_bound}
    Every rectangle in $\mathcal{M}_k$ has size at most
    $k!(k-1)! \cdot C_k^{-1} \cdot 2k \ln k$.
\end{lemma}

There is a natural correspondence between permutations in $\sym_k$ and perfect matchings of a~complete bipartite graph graph $K_{k, k}$. Let us label the vertices of both sides of a bipartition of $K_{k, k}$ as respectively $v_1, \dots, v_k$ and $u_1 \dots, u_k$. For any perfect matching $M \subseteq \{v_1, \dots, v_k\} \times \{u_1, \dots, u_k\}$, its corresponding permutation is $\rho \in \sym_k$ defined as
$$
    \rho(s) = t \quad \text{if and only if} \quad (v_s, u_t) \in M.
$$
It is easy to see that for any two perfect matchings $M_1, M_2$ in $K_{k, k}$, their union $M_1 \cup M_2$ forms a~Hamiltonian cycle iff $\rho_2^{-1} \rho_1 \in \bar\sym_k$ iff $\rho_1^{-1} \rho_2 \in \bar\sym_k$, where $\rho_i$ denotes a permutation corresponding to $M_i$.

\subsection{Complete graphs} \label{sec:mm*}

Similarly to bipartite cliques, every perfect matching in a complete graph $K_{2k}$ can be represented by a permutation in the set $\sym_{2k}^2$. It is then easy to see that the union of two perfect matchings $M_1$, $M_2$ in $K_{2k}$ forms a Hamiltonian cycle in $K_{2k}$ iff $\rho_2 \rho_1$ is in $\sym_{2k}^k$.

\begin{definition}
    A clique matching compatibility matrix $\mathcal{M}^*_k$ of order $k$ is a binary matrix of size $(2k - 1)!! \times (2k - 1)!!$ with rows and columns indexed by permutations from $\sym_{2k}^2$ which satisfies
    $$
        \mathcal{M}^*_k(\rho_1, \rho_2) = 1 \quad \text{if and only if} \quad
        \rho_2 \rho_1 \in \sym_{2k}^k.
    $$
\end{definition}

\begin{lemma}\label{lem:mm*_size_bound}
    Every rectangle in $\mathcal{M}^*_k$ has size at most
    $(2k-1)! \cdot C_k^{-1} \cdot 2k \ln k$.
\end{lemma}

\begin{proof}
    Pick any such rectangle $R$ and let $P_1, P_2$ denote the sets of permutations which are indices of, respectively, rows and columns of $R$. For a set $C \subseteq [2k]$ of size $k$, we define $\sym_C \subseteq \sym_{2k}^2$ as the permutations with all cycles of size $2$ which map all elements of $C$ to $[2k] - C$. Note that $\sym_C = \sym_{[2k] - C}$. Moreover, for every pair $\rho_1 \in P_1, \rho_2 \in P_2$, we have $\rho_1, \rho_2 \in \sym_C$ if we set $C$ to be one of two maximal independent sets of the cycle formed by the union of the edges of $\rho_1, \rho_2$. Therefore
    $$
        P_1 \times P_2 \subseteq \bigcup_{C \cup D = [2k] \atop |C| = |D| = k} \sym_C \times \sym_C,
    $$
    and so
    $$
        |P_1| \cdot |P_2| \leq
        \sum_{C \cup D = [2k] \atop |C| = |D| = k}
            |(P_1 \cap \sym_C) \times (P_2 \cap \sym_C)|.
    $$
    Thus, it suffices to show that $|(P_1 \cap \sym_C) \times (P_2 \cap \sym_C)| \leq k!(k - 1)! \cdot C_k^{-1} \cdot 2k \ln k$ for every $C$.

    Fix $C \cup D = [2k]$, $|C| = |D| = k$, and enumerate $C := \{ c(1), \dots, c(k) \}$ and $D := \{ d(1), \dots, d(k) \}$ (we treat $c$ and $d$ as bijective functions from $[k]$ to resp. $C$ and $D$).
    Put $P_1^* = \{d^{-1} \rho_1 c : \rho_1 \in P_1 \cap \sym_C\}$
    and $P_2^* = \{c^{-1} \rho_2 d : \rho_2 \in P_2 \cap \sym_C \}$. Note that both sets are well defined and $P_1^*, P_2^* \subseteq \sym_k$.
    Since $c, d$ are bijective, we have $|P_1^*| = |P_1 \cap \sym_C|$ and $|P_2^*| = |P_2 \cap \sym_C|$.
    
    Finally, for every $\rho_1^* \in P^*_1, \rho_2^* \in P^*_2$, we have $\rho_2^* \rho_1^* = c^{-1} (\rho_2 \rho_1) c \in \bar\sym_k$. Thus, $P_1^*, P_2^*$ form a~rectangle in $\mathcal{M}_k$, hence by \Cref{lem:mm_size_bound},
    $$
        |(P_1 \cap \sym_C) \times (P_2 \cap \sym_C)| =
        |P_1^* \times P_2^*| \leq k!(k - 1)! \cdot C_k^{-1} \cdot 2k \ln k.
    $$
    Summing over all $C$, we get $|P_1| \cdot |P_2| \leq \binom{2k - 1}{k - 1} \cdot k!(k - 1)! \cdot C_k^{-1} \cdot k \ln k = (2k - 1)! \cdot C_k^{-1} \cdot k \ln k$, which finishes the proof of the lemma.
\end{proof}
\section{Traveling salesperson problem}
\label{sec:tsp}

\subsection{Directed graphs}

Let $G_{n, k}$ denote a directed graph with
$
V(G_{n, k}) = \{ v_{c, r} \colon c \in [n], r \in [k] \}
$
and
$$
E(G_{n, k}) = \{ (v_{c, r_1}, v_{(c \mod n) + 1, r_2}) \mid c \in [n], r_1, r_2 \in [k] \}.
$$

It is easy to see that the pathwidth of $G_{n, k}$ is at most $3k$.
Combining the following lemma with \Cref{lem:raz_spieker} immediately gives \Cref{thm:main_dtsp}. The rest of this section will be dedicated to proving it.

\begin{lemma}\label{lem:main_dtsp}
    For every $k \geq 1$ and $n \geq 3k + 3$, any tropical circuit calculating $DTSP_{G_{n, k}}$ is of size at least $C_k / (2k \ln k)$.
\end{lemma}

\begin{proof}
    Fix $n, k$ and put $G := G_{n, k}$. Let $V_i = \{ v_{i, r} : r \in [k] \}$. Let $G^i = G[V_i \cup V_{(i \mod n) + 1}]$ and let $E_i = E(G^i)$, i.e., $E_i$ contains all edges whose tail belongs to $V_i$.
    For every perfect matching $M$ of $G^i$, we will identify it with a permutation $\rho \in \sym_k$ defined as
    $$
    \rho(s) = t
    \quad \text{iff} \quad
    (v_{i, s}, v_{(i \mod n) + 1, t}) \in M.
    $$
    For every Hamiltonian cycle $H$ of $G$, the set $E(H) \cap E_i$ is a perfect matching in $G^i$. We will say that the sequence of permutations $\rho_1, \dots, \rho_n \in \sym_k$ represents $H$ if $E(H) \cap E_i = \rho_i$ for all $i \in [n]$. Note that such representing set is unique, and moreover, satisfies
    $$
    \rho_i \rho_{i - 1} \dots
    \rho_2 \rho_1
    \rho_n \rho_{n - 1} \dots
    \rho_{i + 2} \rho_{i + 1} \in \bar\sym_k
    $$
    for any $i \in [k]$.
    Conversely, every sequence of permutations satisfying the above represents a unique Hamiltonian cycle of $G$.
    In particular, there are $(k - 1)! \cdot \left(k!\right)^{n - 1}$ such sequences and, hence, Hamiltonian cycles in $G$.
    Additionally, for any set of indices $i_1, \dots, i_{\ell} \in [n]$ of size at most $n - 1$, if we fix $\rho_{i_1}, \dots, \rho_{i_{\ell}}$, then there are exactly $(k - 1)! \cdot (k!)^{n - 1 - \ell}$ ways to fix the rest of the permutations for the above inclusion to hold.
    \newline

    Now, look at the polynomial $DTSP_G$. Let $\bar{E} = \{ (v_{i, 1}, v_{(i \mod n) + 1, 1}) \mid i \in [n] \}$ and let $\bar{X} = \{ x_{s,t} \mid (s, t) \in \bar{E} \}$. That is, the set $\bar{X}$ contains exactly one variable corresponding to an edge in $E_i$ for each $i \in [n]$. Similarly, let $X_i$ denote the set $\{ x_{s,t} \mid  (s, t) \in E_i \}$.

    We will say that a pair of polynomials $g, h$ forms a rectangle $g \cdot h$ in $DTSP_G$ if $g \cdot h \subseteq DTSP_G$, i.e., for every pair of monomials $g' \in g, h' \in h$, the variables of $\sup(g') \cup \sup(h')$ correspond to edges forming a Hamiltonian cycle of $G$.
    Additionally, we will say that such rectangle is balanced if $|\sup(g) \cap \bar{X}|, |\sup(h) \cap \bar{X}| \leq \frac{2}{3} |\bar{X}|$. We would like to prove the following claim.
    \begin{claim}\label{lem:dtsp_rect_bound}
        Let $g \cdot h$ be any balanced rectangle in $DTSP_G$. Then
        $$
        |g \cdot h| \leq (k - 1)! \cdot (k!)^{n - 1} \cdot C_k^{-1} \cdot (2k \ln k).
        $$
    \end{claim}

    First, we finish the proof of the lemma given the claim. Let $\tau$ denote the size of the smallest tropical circuit calculating $DTSP_G$. \Cref{lem:tc_decomposition} says that $DTSP_G$ can be covered by a union of $\tau$ balanced rectangles. Thus,
    $$
        (k - 1)! \cdot (k!)^{n - 1} = |DTSP_G| \leq \tau \cdot (k - 1)! \cdot (k!)^{n - 1} \cdot C_k^{-1} \cdot (2k \ln k),
    $$
    hence
    $
    \tau \geq C_k / (2k \ln k).
    $
    The rest of the section is dedicated to proving the claim.
    \newline

    Fix any balanced rectangle $g \cdot h$. Obviously, supports of $g$ and $h$ are disjoint. We will say that $E_i$ is monochromatic w.r.t. $g \cdot h$ if either $X_i \cap \sup(g)$ or $X_i \cap \sup(h)$ is empty. Let $i_1, \dots, i_\ell$ denote the indices of sets $E_i$ that are not monochromatic. We consider two cases depending on whether $\ell > k$ or not.

    Assume $\ell > k$ and fix a non-monochromatic index $i_j$. Define
    $$
    V^g_{tail} = \{ u \in V_{i_j} \mid \exists_{w \in V_{(i_j \mod n) + 1}} x_{u, w} \in \sup(g) \},
    $$
    $$
    V^g_{head} = \{ w \in V_{(i_j \mod n) + 1} \mid \exists_{u \in V_{i_j}} x_{u, w} \in \sup(g) \},
    $$
    i.e., $V^g_{tail}$ (resp. $V^g_{head}$) denote tails (resp. heads) of all edges in $E_{i_j}$ whose related variables belong to the support of $g$. We define $V^h_{tail}, V^h_{head}$ analogously. It is easy to see that $V^g_{tail} \cap V^h_{tail} = V^g_{head} \cap V^h_{head} = \emptyset$.

    A perfect matching in $G^{i_j}$ induced by any monomial of $g \cdot h$ must match $V^g_{head}$ with $V^g_{tail}$ and $V^h_{head}$ with $V^h_{tail}$, and the number of possible ways to do that is
    $
    |V^g_{head}|! \cdot |V^h_{head}|!
    $, which is at most $(k - 1)!$ as both sets are nonempty and disjoint. Therefore, the number of all cycles corresponding to monomials of $g \cdot h$ is at most
    $$
    ((k - 1)!)^\ell \cdot (k!)^{n - \ell} = \frac{1}{k^\ell} \cdot (k!)^n \leq
    \frac{1}{k^{k + 1}} \cdot (k!)^n = (k - 1)! \cdot (k!)^{n - 1} \cdot \frac{1}{k^k}
    $$
    and $\frac{1}{k^k} \leq C_k^{-1} \leq C_k^{-1} \cdot (2k \ln k)$.
    \newline
    
    Now, assume $\ell \leq k$. Thus, at least $n - k \geq \frac{2}{3}n + 1$ of $E_i$ are monochromatic. Since $g \cdot h$ is balanced, we have two indices $i_g, i_h$ such that $E_{i_g}, E_{i_h}$ are both monochromatic, and both intersections $X_{i_g} \cap \sup(h)$ and $X_{i_h} \cap \sup(g)$ are empty. In particular, the support of every monomial in $g$ corresponds to edges whose intersection with $E_{i_g}$ induces a perfect matching in $G^{i_g}$. The same holds for $h$ and $G^{i_h}$. W.l.o.g. we can assume that $i_g < i_h$. Let $\bar{E} = E(G) - (E_{i_g} \cup E_{i_h})$.

    For a monomial $g' \in g$ (and analogously for $h' \in h$), we define its type as the set
    $$
    \lambda_{g'} := \left\{
        e \in \bar{E} \colon
        x_e \in \sup(g')
    \right\}.
    $$
    Fix an arbitrary pair of types $\lambda_{g^*}, \lambda_{h^*}$ and let $g^*$ (resp. $h^*$) denote the set of all monomials of $g$ (resp. $h$) with that type. The number of such possible pairs is bounded by the number of different projections of a Hamiltonian cycle onto $\bar{E}$, hence is at most $(k!)^{n - 2}$. By the definition, every cycle $H$ corresponding to some monomial of $g^* \cdot h^*$ has the same intersection with $\bar{E}$.

    Let $\rho_i$ for $i \in [n] - \{i_g, i_h\}$ represent the set $(\lambda_{g^*} \cup \lambda_{h^*}) \cap E_i$ which is a perfect matching in $G^i$. Let $\bar\rho_l = \rho_{i_h - 1}\rho_{i_h - 2}\dots\rho_{i_g + 1}$ and $\bar\rho_r = \rho_{i_g-1}\rho_{i_g-2}\dots\rho_1\rho_n\dots\rho_{i_h+1}$.
    Let $P_g$ (resp. $P_h$) denote the set of perfect matchings of $G^{i_g}$ (resp. $G^{i_h})$ induced by the monomials in $g^*$ (resp. $h^*$). Every permutation in the product
    $(\bar\rho_l P_g \bar\rho_r) \cdot P_h$ belongs to $\bar\sym_k$, hence by \Cref{lem:mm_size_bound},
    $$
    k(k - 1)! \cdot C_k^{-1} \cdot (2k \ln k) \geq
    |\bar\rho_l P_g \bar\rho_r| \cdot |P_h| =
    |P_g| \cdot |P_h| = |g^*| \cdot |h^*| \geq |g^* \cdot h^*|.
    $$
    Therefore $|g \cdot h| \leq (k!)^{n - 1} \cdot (k - 1)! \cdot C_k^{-1} \cdot (2k \ln k)$.
    
\end{proof}

\subsection{Undirected graphs}

Let $\bar{G}_{n, k}$ denote an undirected graph with
$
V(\bar{G}_{n, k}) = \{ v_{c,r,i} \colon c \in [n], r \in [k], i \in \{-1, 0, 1\} \}
$ and
\begin{align*}
E(\bar{G}_{n, k}) = \{
    (v_{c, r_1, 1}, v_{(c \mod n) + 1, r_2, -1}) &\mid c \in [n], r_1, r_2 \in [k]
\} \ \cup \\
 \{ (v_{c, r, -1}, v_{c, r, 0}) &\mid c \in [n], r \in [k] \} \ \cup \\
 \{ (v_{c, r, 0}, v_{c, r, 1}) &\mid c \in [n], r \in [k] \}.
\end{align*}
The graph $\bar{G}_{n, k}$ is obtained by performing a textbook reduction on $G_{n, k}$ from directed to undirected version of TSP.

\begin{figure}[h]
    \centering
    \label{fig:reduction}
    \begin{subfigure}{0.45\textwidth}
        \centering        \begin{tikzpicture}
		\node [style=node] (0) at (0, 1) [label=$v_{i,1}$] {};
		\node [style=node] (1) at (0, -1) [label=$v_{i,2}$] {};
		\node [style=none] (2) at (-1.5, 1) {};
		\node [style=none] (3) at (-1.5, 0.25) {};
		\node [style=none] (4) at (-1.5, -0.25) {};
		\node [style=none] (5) at (-1.5, -1) {};
		\node [style=none] (6) at (1.5, -1) {};
		\node [style=none] (7) at (1.5, -0.25) {};
		\node [style=none] (8) at (1.5, 0.25) {};
		\node [style=none] (9) at (1.5, 1) {};
        
		\draw [style=directed] (2.center) to (0);
		\draw [style=directed] (3.center) to (1);
		\draw [style=directed] (4.center) to (0);
		\draw [style=directed] (5.center) to (1);
		\draw [style=directed] (1) to (6.center);
		\draw [style=directed] (1) to (8.center);
		\draw [style=directed] (0) to (7.center);
		\draw [style=directed] (0) to (9.center);
\end{tikzpicture}
        \caption{$i$-th column of $G_{n,2}$}
    \end{subfigure}
    \hfill
    \begin{subfigure}{0.45\textwidth}
        \centering        \begin{tikzpicture}
		\node [style=node] (0) at (0, 1) [label={$v_{i,1,0}$}] {};
		\node [style=node] (1) at (0, -1) [label={$v_{i,2,0}$}] {};
		\node [style=node] (2) at (-1, -1) [label={$v_{i,2,-1}$}] {};
		\node [style=node] (3) at (1, -1) [label={$v_{i,2,1}$}] {};
		\node [style=node] (4) at (1, 1) [label={$v_{i,1,1}$}] {};
		\node [style=node] (5) at (-1, 1) [label={$v_{i,1,-1}$}] {};
		\node [style=none] (6) at (-2.5, 1) {};
		\node [style=none] (7) at (-2.5, 0.25) {};
		\node [style=none] (8) at (-2.5, -0.25) {};
		\node [style=none] (9) at (-2.5, -1) {};
		\node [style=none] (10) at (2.5, -1) {};
		\node [style=none] (11) at (2.5, -0.25) {};
		\node [style=none] (12) at (2.5, 0.25) {};
		\node [style=none] (13) at (2.5, 1) {};
        
		\draw [style=undirected] (6.center) to (5);
		\draw [style=undirected] (7.center) to (2);
		\draw [style=undirected] (8.center) to (5);
		\draw [style=undirected] (9.center) to (2);
		\draw [style=undirected] (5) to (0);
		\draw [style=undirected] (0) to (4);
		\draw [style=undirected] (2) to (1);
		\draw [style=undirected] (1) to (3);
		\draw [style=undirected] (4) to (13.center);
		\draw [style=undirected] (3) to (12.center);
		\draw [style=undirected] (4) to (11.center);
		\draw [style=undirected] (3) to (10.center);
\end{tikzpicture}
        \caption{in $\bar{G}_{n,2}$}
    \end{subfigure}
    \caption{The reduction from $G_{n,k}$ to $\bar{G}_{n,k}$}
\end{figure}
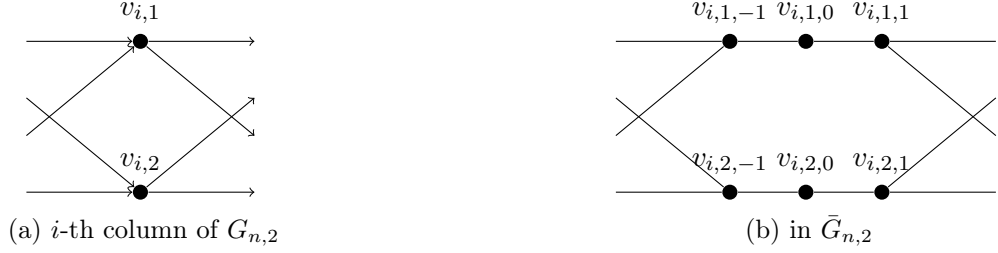

Again, the pathwidth of $\bar{G}_{n, k}$ is at most $3k$.
Combining the following lemma with \Cref{lem:raz_spieker} immediately gives \Cref{thm:main_tsp}.

\begin{lemma}\label{lem:tc_tsp}
    For any $k \geq 1$ and $n > 3k$, any tropical circuit calculating $TSP_{\bar{G}_{n, k}}$ uses at least
    $C_k / (2k \ln k)$
    gates.
\end{lemma}

\begin{proof}
    Let $\bar\Gamma$ be an arbitrary tropical circuit of size $\bar\tau$ computing $TSP_{\bar{G}_{n, k}}$. We will show that there exists a tropical circuit $\Gamma$ of size $\tau \leq \bar\tau$ computing $DTSP_{G_{n, k}}$. In particular, together with \Cref{lem:main_dtsp} this shows that $\bar\tau \geq C_k/(2k \ln k)$.

    Let $\Gamma$ be the tropical circuit obtained by taking $\bar\Gamma$ and performing the following substitutions:
    \begin{itemize}
        \item for every input node labeled with a variable $x_{u_1, u_2}$, where $u_1, u_2 \in V(\bar{G}_{n, k})$, $u_1 = v_{c, r, -1}, u_2 = v_{c, r, 0}$, we replace it with a constant $0$,
        \item for every input node labeled with a variable $x_{u_1, u_2}$, where $u_1, u_2 \in V(\bar{G}_{n, k})$, $u_1 = v_{c, r, 0}, u_2 = v_{c, r, 1}$, we replace it with a constant $0$,
        \item for every input node labeled with a variable $x_{u_1, u_2}$, where $u_1, u_2 \in V(\bar{G}_{n, k})$, $u_1 = v_{c, r_1, 1}, u_2 = v_{(c \mod n) + 1, r_2, -1}$, we replace it with a variable $x_{u_1', u_2'}$ where $u_1', u_2' \in V(G_{n, k})$, $u_1' = v_{c, r_1}, u_2' = v_{(c \mod n) + 1, r_2}$.
    \end{itemize}
    Naturally, the size of $\Gamma$ is at most $\bar\tau$. Let $f$ be some polynomial calculated by $\Gamma$. The reduction works in such a way that there is a bijection between directed Hamiltonian cycles in $G_{n, k}$ and undirected ones in $\bar{G}_{n, k}$, and this bijection directly follows the substitution of edges described in the definition of $\Gamma$. In particular, it can be easily seen that $f \simeq DTSP_{G_{n,k}}$.

\end{proof}
\section{Spanning tree}
\label{sec:st}

Let $H_{n, k}$ denote a directed graph with $V(H_{n,k}) = \{v_{c,r} : c \in [n], r \in [2k]\}$ and
\begin{align*}
    E(H_{n, k}) = \{(v_{c_1, r_1}, v_{c_2, r_2}) &\mid
    c_1, c_2 \in [n], r_1, r_2 \in [2k], |c_1 - c_2| = 1 \} \ \cup \\
    \{(v_{c, r_1}, v_{c, r_2}) &\mid
    c \in \{1, n\}, r_1, r_2 \in [2k], r_1 \neq r_2 \}.
\end{align*}

\begin{figure}[h]
    \centering
    \label{fig:dst}
    \begin{tikzpicture}
		\node [style=node] (0) at (-3.5, 1.5) [label={$v_{1,1}$}] {};
		\node [style=node] (1) at (-3.5, 0.5) [label={$v_{1,2}$}] {};
		\node [style=node] (2) at (-3.5, -0.5) [label={$v_{1,3}$}] {};
		\node [style=node] (3) at (-3.5, -1.5) [label={$v_{1,4}$}] {};
		\node [style=node] (4) at (-2, 1.5) {};
		\node [style=node] (5) at (-2, 0.5) {};
		\node [style=node] (6) at (-2, -0.5) {};
		\node [style=node] (7) at (-2, -1.5) {};
		\node [style=none] (16) at (-0.5, 1.5) {};
		\node [style=none] (17) at (-0.5, 0.5) {};
		\node [style=none] (18) at (-0.5, -0.5) {};
		\node [style=none] (19) at (-0.5, -1.5) {};
		\node [style=none] (20) at (-0.5, -1.25) {};
		\node [style=none] (21) at (-0.5, -0.75) {};
		\node [style=none] (22) at (-0.5, -0.25) {};
		\node [style=none] (23) at (-0.5, 0.25) {};
		\node [style=none] (24) at (-0.5, 0.75) {};
		\node [style=none] (25) at (-0.5, 1.25) {};
		\node [style=node] (26) at (3.5, 1.5) [label={$v_{n,1}$}] {};
		\node [style=node] (27) at (3.5, 0.5) [label={$v_{n,2}$}] {};
		\node [style=node] (28) at (3.5, -0.5) [label={$v_{n,3}$}] {};
		\node [style=node] (29) at (3.5, -1.5) [label={$v_{n,4}$}] {};
		\node [style=node] (30) at (2, 1.5) {};
		\node [style=node] (31) at (2, 0.5) {};
		\node [style=node] (32) at (2, -0.5) {};
		\node [style=node] (33) at (2, -1.5) {};
		\node [style=none] (34) at (0.5, 1.5) {};
		\node [style=none] (35) at (0.5, 0.5) {};
		\node [style=none] (36) at (0.5, -0.5) {};
		\node [style=none] (37) at (0.5, -1.5) {};
		\node [style=none] (38) at (0.5, -1.25) {};
		\node [style=none] (39) at (0.5, -0.75) {};
		\node [style=none] (40) at (0.5, -0.25) {};
		\node [style=none] (41) at (0.5, 0.25) {};
		\node [style=none] (42) at (0.5, 0.75) {};
		\node [style=none] (43) at (0.5, 1.25) {};
		\node [style=none] (44) at (0, 0) {$\dots$};
        
		\draw [style=bi-directional] (0) to (4);
		\draw [style=bi-directional] (0) to (5);
		\draw [style=bi-directional] (0) to (6);
		\draw [style=bi-directional] (0) to (7);
		\draw [style=bi-directional] (1) to (4);
		\draw [style=bi-directional] (1) to (5);
		\draw [style=bi-directional] (1) to (6);
		\draw [style=bi-directional] (1) to (7);
		\draw [style=bi-directional] (2) to (4);
		\draw [style=bi-directional] (2) to (5);
		\draw [style=bi-directional] (2) to (6);
		\draw [style=bi-directional] (2) to (7);
		\draw [style=bi-directional] (3) to (4);
		\draw [style=bi-directional] (3) to (5);
		\draw [style=bi-directional] (3) to (6);
		\draw [style=bi-directional] (3) to (7);
		\draw [style=bi-directional, bend right=90, looseness=1.25] (0) to (3);
		\draw [style=bi-directional, bend right=75, looseness=1.25] (1) to (2);
		\draw [style=bi-directional, bend right=75, looseness=1.25] (0) to (1);
		\draw [style=bi-directional, bend right=75, looseness=1.25] (2) to (3);
		\draw [style=bi-directional, bend right=75, looseness=1.25] (0) to (2);
		\draw [style=bi-directional, bend right=75, looseness=1.25] (1) to (3);
		\draw [style=directed] (16.center) to (4);
		\draw [style=directed] (17.center) to (5);
		\draw [style=directed] (18.center) to (6);
		\draw [style=directed] (19.center) to (7);
		\draw [style=directed] (20.center) to (6);
		\draw [style=directed] (20.center) to (5);
		\draw [style=directed] (20.center) to (4);
		\draw [style=directed] (21.center) to (7);
		\draw [style=directed] (22.center) to (5);
		\draw [style=directed] (22.center) to (4);
		\draw [style=directed] (23.center) to (7);
		\draw [style=directed] (23.center) to (6);
		\draw [style=directed] (24.center) to (4);
		\draw [style=directed] (25.center) to (7);
		\draw [style=directed] (25.center) to (6);
		\draw [style=directed] (25.center) to (5);
		\draw [style=bi-directional] (26) to (30);
		\draw [style=bi-directional] (26) to (31);
		\draw [style=bi-directional] (26) to (32);
		\draw [style=bi-directional] (26) to (33);
		\draw [style=bi-directional] (27) to (30);
		\draw [style=bi-directional] (27) to (31);
		\draw [style=bi-directional] (27) to (32);
		\draw [style=bi-directional] (27) to (33);
		\draw [style=bi-directional] (28) to (30);
		\draw [style=bi-directional] (28) to (31);
		\draw [style=bi-directional] (28) to (32);
		\draw [style=bi-directional] (28) to (33);
		\draw [style=bi-directional] (29) to (30);
		\draw [style=bi-directional] (29) to (31);
		\draw [style=bi-directional] (29) to (32);
		\draw [style=bi-directional] (29) to (33);
		\draw [style=bi-directional, bend left=90, looseness=1.25] (26) to (29);
		\draw [style=bi-directional, bend left=75, looseness=1.25] (27) to (28);
		\draw [style=bi-directional, bend left=75, looseness=1.25] (26) to (27);
		\draw [style=bi-directional, bend left=75, looseness=1.25] (28) to (29);
		\draw [style=bi-directional, bend left=75, looseness=1.25] (26) to (28);
		\draw [style=bi-directional, bend left=75, looseness=1.25] (27) to (29);
		\draw [style=directed] (34.center) to (30);
		\draw [style=directed] (35.center) to (31);
		\draw [style=directed] (36.center) to (32);
		\draw [style=directed] (37.center) to (33);
		\draw [style=directed] (38.center) to (32);
		\draw [style=directed] (38.center) to (31);
		\draw [style=directed] (38.center) to (30);
		\draw [style=directed] (39.center) to (33);
		\draw [style=directed] (40.center) to (31);
		\draw [style=directed] (40.center) to (30);
		\draw [style=directed] (41.center) to (33);
		\draw [style=directed] (41.center) to (32);
		\draw [style=directed] (42.center) to (30);
		\draw [style=directed] (43.center) to (33);
		\draw [style=directed] (43.center) to (32);
		\draw [style=directed] (43.center) to (31);
\end{tikzpicture}
    \caption{$H_{n,2}$}
\end{figure}

It is easy to see that the pathwidth of $H_{n, k}$ is at most $4k$. Combining the following lemma with \Cref{lem:raz_spieker} immediately gives \Cref{thm:main_dst}.

\begin{lemma}\label{lem:tc_dst}
    For any $k \geq 1$ and $n > 12k^2 \ln k + 3$, any tropical circuit calculating $DST_{H_{n, k}}$ uses at least
    $C_k / (2k^2 \ln k)$
    gates.
\end{lemma}

The rest of this section is dedicated to proving this lemma. The majority of the proof follows the same line of argumentation used to prove \Cref{lem:main_dtsp}.
For the rest of the section, we fix suitable $n, k$ and put $G := H_{n, k}$. Let $V_i = \{v_{i, r} : r \in [2k]\}$. For $i \in [n - 1]$, let $G^i = G[V_i \cup V_{i + 1}] - (E(G[V_i]) \cup E(G[V_{i+1}]))$ and let $G^0 = G[V_1], G^n = G[V_n]$. Note that $G^i$ is complete bipartite for $1 \leq i \leq n - 1$ and complete for $i \in \{0, n\}$. We denote $E_i = E(G^i)$ for all $i \in \{0, \dots, n\}$ and put $X_i = \{x_e \colon e \in E_i\}$.

For every perfect matching $M$ of $G^i$ for $i \in \{1, \dots, n-1\}$, we will identify it with a permutation $\rho \in \sym_{2k}$ defined as
    $$
        \rho(s) = t \quad \text{iff} \quad (v_{i, s}, v_{i + 1, t}) \in M \ \text{or}\ (v_{i+1, t}, v_{i,s}) \in M.
    $$
    Similarly, for $i \in \{0, n\}$ we identify any perfect matching $M$ of $G^i$ with $\rho \in \sym_{2k}^2$ such that
    $$
        \rho(s) = t \quad \text{iff} \quad (v_{i, s}, v_{i, t}) \in M \ \text{or}\ (v_{i, t}, v_{i, s}) \in M.
    $$

    We will say that a Hamiltonian cycle of $G$ is \emph{nice} if its intersection with each $E_i$ is a perfect matching in $G^i$ (ignoring edge directions). We will say that a Hamiltonian path of $G$ is \emph{nice} if it is a subgraph of a nice Hamiltonian cycle and its first and last vertices belong to $V_1$. Every nice path is a subgraph of exactly one nice cycle, and conversely, every nice cycle contains exactly $k$ nice paths as a subgraph.

    We define the content of a nice Hamiltonian cycle $H$ as the unique sequence of permutations $\rho_0, \dots, \rho_n \in \sym_{2k}$ such that $E(H) \cap E_i = \rho_i$ for all $0 \leq i \leq n$.
    For the sake of clarity, for such a sequence, we define the following notation.
    $$
        \overleftarrow{\rho_i} := \rho_i \rho_{i - 1} \dots \rho_2 \rho_1
    $$
    $$
        \overrightarrow{\rho_i} := \rho_{n - 1} \rho_{n - 2} \dots \rho_{i + 1} \rho_{i}
    $$
    We assume $\overleftarrow{\rho_0} = \overrightarrow{\rho_n} = \textbf{id}_{\sym_{2k}}$. 
    Note that as described in \Cref{sec:mm*},
    \begin{equation}\label{eq:dst_nice_cycles}
        \left(
            \overleftarrow{\rho_{i - 1}} \cdot
            \rho_0 \cdot
            \overleftarrow{\rho_{i - 1}}^{-1}
        \right)
        \cdot
        \left(
            \overrightarrow{\rho_i}^{-1} \cdot
            \rho_n \cdot
            \overrightarrow{\rho_i}
        \right)
        \in \sym_{2k}^k
    \end{equation}
    for any $i \in [n]$.
    Conversely, for every sequence of permutations satisfying the above (for arbitrarily chosen $i$), there are exactly two nice cycles $H$ such that their content is $\rho_0, \dots, \rho_n$. Both cycles differ only by their direction. In particular, if \Cref{eq:dst_nice_cycles} holds for one such $i$, then it holds for every $i \in [n]$.

    \begin{claim}
        There are $((2k)!)^{n - 1} \cdot (2k - 1)!$ sequences $(\rho_j)_{j=0}^n$ satisfying (1).
    \end{claim}

    \begin{proof}
        We fix $\rho_j$ arbitrarily for all $j \neq i$; there are $((2k - 1)!!)^2 \cdot ((2k)!)^{n - 2}$ ways to do that. Let $\bar\rho = \left(\overrightarrow{\rho_i}^{-1} \cdot \rho_n \cdot \overrightarrow{\rho_i}\right)$, hence we have $\bar\rho = \rho_i^{-1} \tilde\rho \rho_i$, where $\tilde\rho \in \sym_{2k}^2$ is already fixed. By \Cref{prop:sym_k}, the number of values $\bar\rho$ can take for \Cref{eq:dst_nice_cycles} to hold is $(2k - 2)!!$. By the same proposition, the number of $\rho_i \in \sym_{2k}^2$ given $\bar\rho$ is $(2k)!!$. The total number of ways to fix all $\rho_j$ is therefore $((2k - 1)!!)^2 \cdot ((2k)!)^{n - 2} \cdot (2k)!! \cdot (2k - 2)!! = ((2k)!)^{n - 1} \cdot (2k - 1)!$.
    \end{proof}
    
    In particular, the number of nice cycles in $G$ is $2 \cdot ((2k)!)^{n-1} \cdot (2k - 1)!$. By a similar argument, for any set of indices $1 \leq i_1, \dots, i_\ell \leq n - 1$ of size at most $n - 2$, if we fix $\rho_{i_1}, \dots, \rho_{i_\ell}$ together with $\rho_0$ and $\rho_n$, then there are exactly
    $
        ((2k)!)^{n - 2 - \ell} \cdot (2k)!! \cdot (2k - 2)!!
    $
    ways to fix the rest of the permutations for (1) to hold.
    \newline

\begin{proof}[Proof (\Cref{lem:tc_dst})]

    We say that a rectangle $g \cdot h \subseteq DST_G$ covers a nice Hamiltonian cycle $H$ if it contains a monomial encoding a nice Hamiltonian path which is a subgraph of $H$. Since all nice Hamiltonian paths are spanning out-trees as well, any rectangle decomposition of the polynomial calculated by our circuit must cover all nice cycles. We define rectangles being balanced as in the statement of \Cref{lem:tc_decomposition} with respect to the set $X = \{ x_{(v_{c, 1}, v_{c + 1, 1})} \mid c \in [n - 1] \}$.
    
    \begin{claim}\label{clm:dst_rect_bound}
        Let $g \cdot h$ be any balanced rectangle in $DST_G$. Then, the number of nice Hamiltonian cycles covered by $g \cdot h$ is at most
        $$
            ((2k)!)^n \cdot C_k^{-1} \cdot 2k \ln k.
        $$
    \end{claim}

    For any tropical circuit calculating $DST_G$ using $\tau$ gates, \Cref{lem:tc_decomposition} implies that $DST_G$ can be covered by a union of $\tau$ balanced rectangles. Thus, in order to cover every nice cycle, at least one balanced rectangle must cover at least $\frac{2 ((2k)!)^{n - 1} (2k - 1)!}{\tau}$ of those, hence assuming \Cref{clm:dst_rect_bound},
    $$
        \tau \geq C_k / (2k^2 \ln k).
    $$    
    
    The rest of the section is dedicated to proving \Cref{clm:dst_rect_bound}.
    Fix any balanced rectangle $g \cdot h$. The supports of $g$ and $h$ are disjoint. We will say that $E_i$ is monochromatic w.r.t. $g \cdot h$ if either $X_i \cap \sup(g)$ or $X_i \cap \sup(h)$ is empty.

    \begin{claim}\label{clm:dst_rect_colorful}
        Pick $1 \leq i \leq n - 1$. If $E_i$ is not monochromatic w.r.t. $g \cdot h$, then there exists an edge $(u, v) \in E_i$ such that $x_{u, v} \not\in \sup(g) \cup \sup(h)$.
    \end{claim}

    \begin{proof}
        Assume by contradiction that $x_{u, v} \in \sup(g)$ or $x_{u, v} \in \sup(h)$ for every $(u, v) \in E_i$.
        Since all monomials in $g \cdot h$ represent out-trees, for any pair of edges of the form $(u, v), (w, v) \in E_i$, either both $x_{u, v}, x_{w, v}$ belong to $\sup(g)$, or both belong to $\sup(h)$. Similarly, for any $(u, v) \in E_i$ both $x_{u, v}, x_{v, u}$ must also belong to the same support.
        
        By our assumption, we have two edges $(p, q), (s, t) \in E_i$ such that $x_{p, q} \in \sup(g)$ and $x_{s, t} \in \sup(h)$ (and hence $x_{t, s} \in \sup(h)$ as well). Either $(q, s)$ or $(q, t)$ belong to $E_i$. In the former case, we have $x_{q, s} \in \sup(h)$ and $x_{s, q} \in \sup(g)$. In the latter, $x_{q, t} \in \sup(h)$ and $x_{t, q} \in \sup(g)$. Both cases arrive at a contradiction, which proves the claim.
    \end{proof}

    Let $i_1 < \dots < i_\ell$ denote the indices in range $[n - 1]$ of sets $E_i$ which are not monochromatic. We consider two cases depending on whether $\ell > 4k^2 \ln k$ or not.

    Assume $\ell > 4k^2 \ln k$. By sacrificing at most half of indices $i_j$, we can assume that they are non-adjacent. That is, we assume that $\ell > 2k^2 \ln k$ and that $i_{j + 1} - i_j > 1$ for all $j \in [\ell - 1]$. For each $i \in \{i_j\}_{j=1}^\ell$, let $(p_i, q_i)$ denote the edge of $E_i$ given by \Cref{clm:dst_rect_colorful}.
    Consider the following randomized procedure to generate a nice directed Hamiltonian cycle $H$ of $G$.
    \begin{itemize}
        \item Select $\rho_0 \in \sym_{2k}^2$ uniformly randomly.
        \item Select a subset of edges $F_0 \subseteq E_0$ by matching vertices of $V_1$ according to $\rho_0$ and picking the edge directions uniformly and independently. Let $V_1^{tail}$ and $V_1^{head}$ denote the vertices of $V_1$ composed of respectively tails and heads of the edges of $F_0$.
        \item For $i \in \{1, \dots, n - 1\}$ do the following. Pick $\rho_i \in \sym_{2k}$ uniformly randomly and select $F_i \subseteq E_i$ by matching vertices of $V(G^i)$ according to $\rho_i$ and picking the edge directions so that they agree with edges of $F_{i - 1}$ (i.e., if we match $v_{i, r_1}$ with $v_{i + 1, r_2}$ and $v_{i, r_1} \in V_i^{tail}$, then we add $(v_{i + 1, r_2}, v_{i, r_1})$ to $F_i$, and $(v_{i, r_1}, v_{i + 1, r_2})$ otherwise). Define $V_{i + 1}^{tail}, V_{i + 1}^{head}$ analogously.
        \item The union $F_0 \cup \dots \cup F_{n - 1}$ at this point forms $k$ disjoint directed paths whose endpoints are exactly $V_n$. There are exactly $(k - 1)!$ ways to pick $F_n$ in a way that connects these paths into a Hamiltonian cycle. We pick one of such ways uniformly randomly.
    \end{itemize}
    It is easy to see that this procedure generates all nice Hamiltonian cycles of $G$ with uniform probability distribution. Now, we can bound the probability that a cycle generated this way does not contain any of the forbidden edges $(p_i, q_i)$ given by \Cref{clm:dst_rect_colorful}. Note that only such cycles can be covered by $g \cdot h$.

    The steps during which the number of feasible choices could get restricted is when selecting $F_{i_j}$ for some $j \in [\ell]$. Assume w.l.o.g. $p_{i_j} \in V_{i_j}, q_{i_j} \in V_{i_j + 1}$, the other case is symmetric. If $p_{i_j} \in V_{i_j}^{tail}$, then every choice of $F_{i_j}$ will be feasible, as an edge between $p_{i_j}, q_{i_j}$ potentially added will always be directed from $q_{i_j}$ to $p_{i_j}$. In case $p_{i_j} \in V_{i_j}^{head}$, there is exactly $\frac{1}{k}$ chance of adding the forbidden edge to $F_{i_j}$, as $\rho_{i_j}$ is selected uniformly from whole $\sym_{2k}$.

    The crucial observation is that the probability that $p_{i_j}$ is in $V_{i_j}^{head}$ at this point, assuming so far we did not pick any of the forbidden edges, is exactly $1/2$, independent of previous choices. This follows from the fact that $\rho_{i_j - 1}$ was chosen uniformly from $\sym_{2k}$ and due to our assumption that $i_j > i_{j - 1} + 1$. The partition of $V_{i_j - 1}$ into $V_{i_j - 1}^{head}$ and $V_{i_j - 1}^{tail}$ can have arbitrary distribution, however, both sets are always of size $k$, hence the probability of $p_{i_j}$ getting matched by $\rho_{i_j - 1}$ with a vertex from $V_{i_j - 1}^{tail}$ is exactly $\frac{1}{2}$.

    This means, that the probability of a nice Hamiltonian cycle containing none of the forbidden edges is at most
    $$
    \left( 1 - \frac{1}{2k} \right)^\ell < e^{-k \ln k} \leq C_k^{-1} \leq C_k^{-1} \cdot 2k^2 \ln k,
    $$
    hence the number of such cycles is at most
    $$
    \left( 2 ((2k)!)^{n-1} (2k - 1)! \right) 
        \cdot \left( C_k^{-1} \cdot 2k^2 \ln k \right) =
    ((2k)!)^n \cdot C_k^{-1} \cdot 2k \ln k,
    $$
    which finishes the proof in case $\ell > 4k^2 \ln k$.
    \newline

    Now, assume $\ell < 4k^2 \ln k$. Thus, at least $n - 4k^2 \ln k \geq \frac{2n}{3} + 1$ of $E_i$ are monochromatic. Since $g \cdot h$ is balanced, we have two indices $i_g, i_h \in [n - 1]$ s.t. $E_{i_g}, E_{i_h}$ are both monochromatic, and both $X_{i_g} \cap \sup(h)$ and  $X_{i_h} \cap \sup(g)$ are empty. W.l.o.g. we can assume that $i_g < i_h$. Let $\bar{E} = E(G) - (E_{i_g} \cup E_{i_h})$.

    For a monomial $g' \in g$ (and analogously for $h' \in h$), we define its type as the set
    $$
        \lambda_{g'} := \{ e \in \bar{E} : x_e \in \sup(g') \}.
    $$
    Fix an arbitrary pair of types $\lambda_{g^*}, \lambda_{h^*}$ and let $g^*$ (resp. $h^*$) denote the set of all monomials of $g$ (resp. $h$) with that type. By the definition, every spanning out-tree corresponding to some monomial of $g^* \cdot h^*$ has the same intersection with $\bar{E}$. If $g^* \cdot h^*$ covers at least one nice cycle, this intersection must induce a perfect matching in $G^i$ for each $i \in [n] - \{i_g, i_h\}$ and a matching of size $k - 1$ in $G^0$.
    Let $(\rho_i)_{i \in [n] - \{i_g, i_h\}}$ denote the set $(\lambda_{g^*} \cup \lambda_{h^*}) \cap E_i$. Let $\rho_0$ denote the unique permutation of $\sym_{2k}^2$ containing $(\lambda_{g^*} \cup \lambda_{h^*}) \cap E_0$ as a subgraph.
    
    Let $\bar\rho_l = \overleftarrow{\rho_{i_g - 1}}$, $\bar\rho_r = \overrightarrow{\rho_{i_h+1}}$ and $\bar\rho_m = \rho_{i_h - 1}\rho_{i_h - 2}\dots\rho_{i_g + 2}\rho_{i_g + 1}$. Let $P_g$ (resp. $P_h$) denote the set of perfect matchings of $G^{i_g}$ (resp. $G^{i_h})$ induced by the monomials in $g^*$ (resp. $h^*$). Finally, put
    $$
        P_g^* = \left\{
            \rho \cdot
            \bar\rho_l \cdot 
            \rho_0 \cdot 
            \bar\rho_l^{-1} \cdot 
            \rho^{-1}
            \mid \rho \in P_g
        \right\} \ \text{and}
    $$
    $$
        P_h^* = \left\{
            \bar\rho_m^{-1} \cdot
            \rho^{-1} \cdot
            \bar\rho_r^{-1} \cdot 
            \rho_n \cdot 
            \bar\rho_r \cdot 
            \rho \cdot
            \bar\rho_m
            \mid \rho \in P_h
        \right\}.
    $$
    By \Cref{prop:sym_k}, we have $|P_g| \leq (2k)!! \cdot |P_g^*|$ and $|P_h| \leq (2k)!! \cdot |P_h^*|$.
    By \Cref{eq:dst_nice_cycles}, every pair $\rho_g \in P_g^*, \rho_h \in P_h^*$ satisfies $\rho_h^* \rho_g^* \in \sym_{2k}^k$, hence applying \Cref{lem:mm*_size_bound} gives us
    $$
        2 \cdot |P_g| \cdot |P_h| \leq
        2 ((2k)!!)^2 \cdot |P_g^*| \cdot |P_h^*| \leq
        2 ((2k)!!)^2 (2k - 1)! \cdot C_k^{-1} \cdot 2k \ln k.
    $$
    Finally, note that the quantity $2 \cdot |P_g| \cdot |P_h|$ bounds from above the number of nice cycles covered by $g^* \cdot h^*$.

    The number of pairs of types $\lambda_{g^*}, \lambda_{h^*}$ which cover at least one nice cycle is bounded by the number of different projections of a nice path onto $\bar{E}$, hence by
    $$
    k ((2k - 1)!!)^2 ((2k)!)^{n - 3}.
    $$
    Therefore, the total number of nice cycles covered by $g \cdot h$ is at most
    $$
    \left( k ((2k - 1)!!)^2 ((2k)!)^{n - 3} \right) \cdot
    \left( 2 ((2k)!!)^2 (2k - 1)! \cdot C_k^{-1} \cdot 2k \ln k \right) =
    $$
    $$
    = ((2k)!)^n \cdot C_k^{-1} \cdot 2k \ln k,
    $$
    which finishes the proof of the claim.

\end{proof}

\section{Conclusion}
\label{sec:con}

As noted earlier, the exact nondeterministic communication complexity of the matchings compatibility matrix is still open: While $C_k =k^{O(k)}$ holds trivially, the currently best known lower bound is $C_k = 2^{\Omega(k \log \log k)}$~\cite{raz1995log}. 
Proving this lower bound would automatically imply the same bound for \Cref{thm:main_dtsp}, \Cref{thm:main_tsp}, and \Cref{thm:main_dst}. In particular, this complexity would be asymptotically tight (up to factors polynomial in $N$), as tropical circuits of size $2^{\Oh(k \log k)} \cdot N^{\Oh(1)}$ can be easily constructed for all these problems by following algorithms based on naive dynamic programming on path decompositions.
In the reverse direction, better upper bounds on $C_k$ seem also useful for obtaining better tropical circuits for $DTSP_G,TSP_G$ and $DST_G$, and possibly even for obtaining faster algorithm for TSP parameterized by pathwidth. It seems however, as also expressed by the authors of~\cite{raz1995log}, that $C_k$ is closer to $2^{\Omega(k \log k)}$ than it is to $2^{\Theta(k \log \log k)}$.

Another natural opportunity for further work would be to find lower bounds for other computational problems, for example it seems plausible that with tools from Theorem~\ref{thm:main_is} and the reduction ideas from~\cite{LokshtanovMS18b} one can also obtain a $\Omega(3^k)$ lower bound for tropical circuits calculating the minimum weight dominating set of a graph with pathwidth $k$.

More ambitiously, a natural open question is whether the lower bounds on tropical circuits as defined in this paper can be generalized to more expressive variants of tropical circuits (see the book by Jukna~\cite{Jukna2023}).

\bibliographystyle{alpha}
\bibliography{ref}

@Inbook{Jukna2023,
author="Jukna, Stasys",
title="Tropical Circuit Complexity: Limits of Pure Dynamic Programming",
year="2023",
publisher="Springer International Publishing",
address="Cham",
isbn="978-3-031-42354-3",
doi="10.1007/978-3-031-42354-3",
url="https://doi.org/10.1007/978-3-031-42354-3"
}

@article{raz1995log,
  title={On the “log rank”-conjecture in communication complexity},
  author={Raz, Ran and Spieker, Boris},
  journal={Combinatorica},
  volume={15},
  number={4},
  pages={567--588},
  year={1995},
  publisher={Springer}
}

@article{CyganNPPRW22,
  author       = {Marek Cygan and
                  Jesper Nederlof and
                  Marcin Pilipczuk and
                  Michal Pilipczuk and
                  Johan M. M. van Rooij and
                  Jakub Onufry Wojtaszczyk},
  title        = {Solving Connectivity Problems Parameterized by Treewidth in Single
                  Exponential Time},
  journal      = {{ACM} Trans. Algorithms},
  volume       = {18},
  number       = {2},
  pages        = {17:1--17:31},
  year         = {2022},
  url          = {https://doi.org/10.1145/3506707},
  doi          = {10.1145/3506707},
  bibsource    = {dblp computer science bibliography, https://dblp.org}
}

@article{BodlaenderCKN15,
  author       = {Hans L. Bodlaender and
                  Marek Cygan and
                  Stefan Kratsch and
                  Jesper Nederlof},
  title        = {Deterministic single exponential time algorithms for connectivity
                  problems parameterized by treewidth},
  journal      = {Inf. Comput.},
  volume       = {243},
  pages        = {86--111},
  year         = {2015},
  url          = {https://doi.org/10.1016/j.ic.2014.12.008},
  doi          = {10.1016/J.IC.2014.12.008},
  bibsource    = {dblp computer science bibliography, https://dblp.org}
}

@inproceedings{Korhonen21,
  author       = {Tuukka Korhonen},
  editor       = {Nikhil Bansal and
                  Emanuela Merelli and
                  James Worrell},
  title        = {Lower Bounds on Dynamic Programming for Maximum Weight Independent
                  Set},
  booktitle    = {48th International Colloquium on Automata, Languages, and Programming,
                  {ICALP} 2021, Glasgow, Scotland (Virtual Conference), July 12-16,
                  2021},
  series       = {LIPIcs},
  volume       = {198},
  pages        = {87:1--87:14},
  publisher    = {Schloss Dagstuhl - Leibniz-Zentrum f{\"{u}}r Informatik},
  year         = {2021},
  url          = {https://doi.org/10.4230/LIPIcs.ICALP.2021.87},
  doi          = {10.4230/LIPICS.ICALP.2021.87},
  bibsource    = {dblp computer science bibliography, https://dblp.org}
}

@article{JerrumS82,
  author       = {Mark Jerrum and
                  Marc Snir},
  title        = {Some Exact Complexity Results for Straight-Line Computations over
                  Semirings},
  journal      = {J. {ACM}},
  volume       = {29},
  number       = {3},
  pages        = {874--897},
  year         = {1982},
  url          = {https://doi.org/10.1145/322326.322341},
  doi          = {10.1145/322326.322341},
  bibsource    = {dblp computer science bibliography, https://dblp.org}
}

@article{Rothvoss17,
  author       = {Thomas Rothvoss},
  title        = {The Matching Polytope has Exponential Extension Complexity},
  journal      = {J. {ACM}},
  volume       = {64},
  number       = {6},
  pages        = {41:1--41:19},
  year         = {2017},
  url          = {https://doi.org/10.1145/3127497},
  doi          = {10.1145/3127497},
  bibsource    = {dblp computer science bibliography, https://dblp.org}
}

@article{FioriniMPTW15,
  author       = {Samuel Fiorini and
                  Serge Massar and
                  Sebastian Pokutta and
                  Hans Raj Tiwary and
                  Ronald de Wolf},
  title        = {Exponential Lower Bounds for Polytopes in Combinatorial Optimization},
  journal      = {J. {ACM}},
  volume       = {62},
  number       = {2},
  pages        = {17:1--17:23},
  year         = {2015},
  url          = {https://doi.org/10.1145/2716307},
  doi          = {10.1145/2716307},
  bibsource    = {dblp computer science bibliography, https://dblp.org}
}

@article{BodlaenderDFH09,
  author       = {Hans L. Bodlaender and
                  Rodney G. Downey and
                  Michael R. Fellows and
                  Danny Hermelin},
  title        = {On problems without polynomial kernels},
  journal      = {J. Comput. Syst. Sci.},
  volume       = {75},
  number       = {8},
  pages        = {423--434},
  year         = {2009},
  url          = {https://doi.org/10.1016/j.jcss.2009.04.001},
  doi          = {10.1016/J.JCSS.2009.04.001},
  bibsource    = {dblp computer science bibliography, https://dblp.org}
}

@article{FortnowS11,
  author       = {Lance Fortnow and
                  Rahul Santhanam},
  title        = {Infeasibility of instance compression and succinct PCPs for {NP}},
  journal      = {J. Comput. Syst. Sci.},
  volume       = {77},
  number       = {1},
  pages        = {91--106},
  year         = {2011},
  url          = {https://doi.org/10.1016/j.jcss.2010.06.007},
  doi          = {10.1016/J.JCSS.2010.06.007},
  bibsource    = {dblp computer science bibliography, https://dblp.org}
}

@inproceedings{BeckI13,
  author       = {Christopher Beck and
                  Russell Impagliazzo},
  editor       = {Dan Boneh and
                  Tim Roughgarden and
                  Joan Feigenbaum},
  title        = {Strong {ETH} holds for regular resolution},
  booktitle    = {Symposium on Theory of Computing Conference, STOC'13, Palo Alto, CA,
                  USA, June 1-4, 2013},
  pages        = {487--494},
  publisher    = {{ACM}},
  year         = {2013},
  url          = {https://doi.org/10.1145/2488608.2488669},
  doi          = {10.1145/2488608.2488669},
  bibsource    = {dblp computer science bibliography, https://dblp.org}
}

@article{WilliamsW18,
  author       = {Virginia {Vassilevska Williams} and
                  R. Ryan Williams},
  title        = {Subcubic Equivalences Between Path, Matrix, and Triangle Problems},
  journal      = {J. {ACM}},
  volume       = {65},
  number       = {5},
  pages        = {27:1--27:38},
  year         = {2018},
  url          = {https://doi.org/10.1145/3186893},
  doi          = {10.1145/3186893},
  bibsource    = {dblp computer science bibliography, https://dblp.org}
}

@article{KarpUW86,
  author       = {Richard M. Karp and
                  Eli Upfal and
                  Avi Wigderson},
  title        = {Constructing a perfect matching is in random {NC}},
  journal      = {Comb.},
  volume       = {6},
  number       = {1},
  pages        = {35--48},
  year         = {1986},
  url          = {https://doi.org/10.1007/BF02579407},
  doi          = {10.1007/BF02579407},
  bibsource    = {dblp computer science bibliography, https://dblp.org}
}

@article{MulmuleyVV87,
  author       = {Ketan Mulmuley and
                  Umesh V. Vazirani and
                  Vijay V. Vazirani},
  title        = {Matching is as easy as matrix inversion},
  journal      = {Comb.},
  volume       = {7},
  number       = {1},
  pages        = {105--113},
  year         = {1987},
  url          = {https://doi.org/10.1007/BF02579206},
  doi          = {10.1007/BF02579206},
  bibsource    = {dblp computer science bibliography, https://dblp.org}
}

@article{Bjorklund14,
  author       = {Andreas Bj{\"{o}}rklund},
  title        = {Determinant Sums for Undirected Hamiltonicity},
  journal      = {{SIAM} J. Comput.},
  volume       = {43},
  number       = {1},
  pages        = {280--299},
  year         = {2014},
  url          = {https://doi.org/10.1137/110839229},
  doi          = {10.1137/110839229},
  bibsource    = {dblp computer science bibliography, https://dblp.org}
}

@article{calbet2023k_r,
  title={${K_r}$-saturated Graphs and the Two Families Theorem},
  author={Calbet, Asier},
  journal={arXiv preprint arXiv:2302.13389},
  year={2023}
}

@article{LovettA14,
  author       = {Shachar Lovett},
  title        = {Recent Advances on the Log-Rank Conjecture in Communication Complexity},
  journal      = {Bull. {EATCS}},
  volume       = {112},
  year         = {2014},
  url          = {http://eatcs.org/beatcs/index.php/beatcs/article/view/260},
  bibsource    = {dblp computer science bibliography, https://dblp.org}
}

@article{FominLPS16,
  author       = {Fedor V. Fomin and
                  Daniel Lokshtanov and
                  Fahad Panolan and
                  Saket Saurabh},
  title        = {Efficient Computation of Representative Families with Applications
                  in Parameterized and Exact Algorithms},
  journal      = {J. {ACM}},
  volume       = {63},
  number       = {4},
  pages        = {29:1--29:60},
  year         = {2016},
  url          = {https://doi.org/10.1145/2886094},
  doi          = {10.1145/2886094},
  bibsource    = {dblp computer science bibliography, https://dblp.org}
}

@inproceedings{Nederlof20,
  author       = {Jesper Nederlof},
  editor       = {Konstantin Makarychev and
                  Yury Makarychev and
                  Madhur Tulsiani and
                  Gautam Kamath and
                  Julia Chuzhoy},
  title        = {Bipartite {TSP} in $o(1.9999^n)$ time, assuming quadratic
                  time matrix multiplication},
  booktitle    = {Proceedings of the 52nd Annual {ACM} {SIGACT} Symposium on Theory
                  of Computing, {STOC} 2020, Chicago, IL, USA, June 22-26, 2020},
  pages        = {40--53},
  publisher    = {{ACM}},
  year         = {2020},
  url          = {https://doi.org/10.1145/3357713.3384264},
  doi          = {10.1145/3357713.3384264},
  bibsource    = {dblp computer science bibliography, https://dblp.org}
}

@techreport{kerr1970effect,
  title={The effect of algebraic structure on the computational complexity of matrix multiplication},
  author={Kerr, Leslie Robert},
  year={1970},
  institution={Cornell University}
}

@article{Bellman62,
	author    = {Richard Bellman},
	title     = {Dynamic Programming Treatment of the Travelling Salesman Problem},
	journal   = {J. {ACM}},
	volume    = {9},
	number    = {1},
	pages     = {61--63},
	year      = {1962},
	url       = {https://doi.org/10.1145/321105.321111},
	doi       = {10.1145/321105.321111},
	bibsource = {dblp computer science bibliography, https://dblp.org}
}

@article{heldKarp,
	author = {Held, Michael. and Karp, Richard M.},
	title = {A Dynamic Programming Approach to Sequencing Problems},
	journal = {Journal of the Society for Industrial and Applied Mathematics},
	volume = {10},
	number = {1},
	pages = {196-210},
	year = {1962},
	doi = {10.1137/0110015},
	URL = { https://doi.org/10.1137/0110015	},
	eprint = { 	https://doi.org/10.1137/0110015	}
}

@article{Warshall62,
  author       = {Stephen Warshall},
  title        = {A Theorem on Boolean Matrices},
  journal      = {J. {ACM}},
  volume       = {9},
  number       = {1},
  pages        = {11--12},
  year         = {1962},
  url          = {https://doi.org/10.1145/321105.321107},
  doi          = {10.1145/321105.321107},
  bibsource    = {dblp computer science bibliography, https://dblp.org}
}

@article{Floyd62a,
  author       = {Robert W. Floyd},
  title        = {Algorithm 97: Shortest path},
  journal      = {Commun. {ACM}},
  volume       = {5},
  number       = {6},
  pages        = {345},
  year         = {1962},
  url          = {https://doi.org/10.1145/367766.368168},
  doi          = {10.1145/367766.368168},
  bibsource    = {dblp computer science bibliography, https://dblp.org}
}

@book{CyganFKLMPPS15,
  author       = {Marek Cygan and
                  Fedor V. Fomin and
                  Lukasz Kowalik and
                  Daniel Lokshtanov and
                  D{\'{a}}niel Marx and
                  Marcin Pilipczuk and
                  Michal Pilipczuk and
                  Saket Saurabh},
  title        = {Parameterized Algorithms},
  publisher    = {Springer},
  year         = {2015},
  url          = {https://doi.org/10.1007/978-3-319-21275-3},
  doi          = {10.1007/978-3-319-21275-3},
  isbn         = {978-3-319-21274-6},
  bibsource    = {dblp computer science bibliography, https://dblp.org}
}

@article{LokshtanovMS18a,
  author       = {Daniel Lokshtanov and
                  D{\'{a}}niel Marx and
                  Saket Saurabh},
  title        = {Slightly Superexponential Parameterized Problems},
  journal      = {{SIAM} J. Comput.},
  volume       = {47},
  number       = {3},
  pages        = {675--702},
  year         = {2018},
  url          = {https://doi.org/10.1137/16M1104834},
  doi          = {10.1137/16M1104834},
  bibsource    = {dblp computer science bibliography, https://dblp.org}
}

@article{LokshtanovMS18b,
  author       = {Daniel Lokshtanov and
                  D{\'{a}}niel Marx and
                  Saket Saurabh},
  title        = {Known Algorithms on Graphs of Bounded Treewidth Are Probably Optimal},
  journal      = {{ACM} Trans. Algorithms},
  volume       = {14},
  number       = {2},
  pages        = {13:1--13:30},
  year         = {2018},
  url          = {https://doi.org/10.1145/3170442},
  doi          = {10.1145/3170442},
  bibsource    = {dblp computer science bibliography, https://dblp.org}
}

@inproceedings{CurticapeanLN18,
  author       = {Radu Curticapean and
                  Nathan Lindzey and
                  Jesper Nederlof},
  editor       = {Artur Czumaj},
  title        = {A Tight Lower Bound for Counting Hamiltonian Cycles via Matrix Rank},
  booktitle    = {Proceedings of the Twenty-Ninth Annual {ACM-SIAM} Symposium on Discrete
                  Algorithms, {SODA} 2018, New Orleans, LA, USA, January 7-10, 2018},
  pages        = {1080--1099},
  publisher    = {{SIAM}},
  year         = {2018},
  url          = {https://doi.org/10.1137/1.9781611975031.70},
  doi          = {10.1137/1.9781611975031.70},
  bibsource    = {dblp computer science bibliography, https://dblp.org}
}

@article{CyganKN18,
  author       = {Marek Cygan and
                  Stefan Kratsch and
                  Jesper Nederlof},
  title        = {Fast Hamiltonicity Checking Via Bases of Perfect Matchings},
  journal      = {J. {ACM}},
  volume       = {65},
  number       = {3},
  pages        = {12:1--12:46},
  year         = {2018},
  url          = {https://doi.org/10.1145/3148227},
  doi          = {10.1145/3148227},
  bibsource    = {dblp computer science bibliography, https://dblp.org}
}

\end{document}